\documentclass[12pt,a4paper,twoside]{amsart} 
\usepackage{latexsym}
\usepackage{amsmath}
\usepackage{amssymb}
\usepackage{amsfonts}
\usepackage{ifthen}
\usepackage{mathrsfs}               
\usepackage{bbm}                    
\usepackage{bbold}                  
\usepackage{textcomp}               
\usepackage{amstext}
\usepackage{enumerate}
\usepackage{amstext}
\usepackage{units}
\newtheorem{theorem}{Theorem}[section]
\newtheorem{lemma}[theorem]{Lemma}
\newtheorem{corollary}[theorem]{Corollary}
\newtheorem{proposition}[theorem]{Proposition}
\theoremstyle{remark}
\newtheorem{remark}[theorem]{Remark}

\numberwithin{equation}{section}
\newcommand{\CC}{\mathbb{C}}

\newcommand{\NN}{\mathbb{N}}

\newcommand{\RR}{\mathbb{R}}

\newcommand{\ZZ}{\mathbb{Z}}
\newcommand{\supp}{\mathrm{supp}}

\newcommand{\Ran}{\mathrm{Ran}}

\newcommand{\id}{\mathbbm{1}}
\newcommand{\klg}{\leqslant} 
\newcommand{\grg}{\geqslant}          
\newcommand{\ve}{\varepsilon}
\newcommand{\vp}{\varphi}

\newcommand{\vr}{\varrho}
\newcommand{\vt}{\vartheta}
\newcommand{\vs}{\varsigma}

\newcommand{\wt}[1]{\widetilde{#1}}
\newcommand{\SL}{\langle}                          
\newcommand{\SR}{\rangle}    
\newcommand{\SPn}[2]{\langle \,#1\,|\,#2\, \rangle} 
\newcommand{\SPb}[2]{\big\langle \,#1\,\big|\,#2\, \big\rangle} 

\newcommand{\nf}[2]{\nicefrac{#1}{#2}}
\newcommand{\ol}[1]{\overline{#1}} 
\newcommand{\ul}[1]{\underline{#1}}
\newcommand{\bigO}{\mathcal{O}}    
\newcommand{\V}[1]{\mathbf{#1}}
\newcommand{\valpha}{\boldsymbol{\alpha}}
\newcommand{\vxi}{\boldsymbol{\xi}}

\newcommand{\vsigma}{\boldsymbol{\sigma}}

\newcommand{\veps}{\boldsymbol{\varepsilon}}

\newcommand{\LO}{\mathscr{L}}      
\newcommand{\HP}{\mathscr{K}}
\newcommand{\HR}{\mathscr{H}}
\newcommand{\Fock}{\mathscr{F}_{\mathrm{b}}}
\newcommand{\core}{\mathscr{D}}
\newcommand{\dom}{\mathcal{D}}
\newcommand{\form}{\mathcal{Q}}
\newcommand{\spec}{{\sigma}}
\newcommand{\PAm}{P^-_\mathbf{A}}               
\newcommand{\PA}{P^+_{\mathbf{A}}}
\newcommand{\PO}{P^+_{\mathbf{0}}}
\newcommand{\PApm}{P^\pm_{\mathbf{A}}}
\newcommand{\SA}{S_{\mathbf{A}}}  

\newcommand{\DA}{D_{\mathbf{A}}}                
\newcommand{\DO}{D_{\mathbf{0}}}

\newcommand{\RA}[1]{R_{\mathbf{A}}(#1)}
\newcommand{\RAF}[1]{R_{\mathbf{A}}^F(#1)}
\newcommand{\RAmF}[1]{R_{\mathbf{A}}^{-F}(#1)}
\newcommand{\RO}[1]{R_{\mathbf{0}}(#1)}
\newcommand{\Hf}{H_{\mathrm{f}}}                   
\newcommand{\HT}{\check{H}_{\mathrm{f}}}           
\newcommand{\ad}{a^\dagger}                     
\newcommand{\NP}[1]{H^+_{#1}}
\newcommand{\NPneg}[1]{H^-_{#1}}
\newcommand{\FNP}[1]{H^{\mathrm{np}}_{#1}}
\newcommand{\PF}[1]{H^{\mathrm{sr}}_{#1}}
\newcommand{\gc}{\gamma_{\mathrm{c}}} 
\newcommand{\gcnp}{\gamma_{\mathrm{c}}^{\mathrm{np}}}
\newcommand{\gcPF}{\gamma_{\mathrm{c}}^{\mathrm{sr}}}
\newcommand{\UV}{\Lambda}             
\newcommand{\Th}{\Sigma}              
\newcommand{\Thnp}{\Sigma^{\mathrm{np}}}
\newcommand{\ThPF}{\Sigma^{\mathrm{sr}}}
\newcommand{\np}{\mathrm{np}}
\newcommand{\Pf}{\mathrm{sr}}
\newcommand{\el}{\mathrm{el}}
\newcommand{\cR}{\mathcal{R}}
\newcommand{\cS}{\mathcal{S}}
\newcommand{\cT}{\mathcal{T}}

\newcommand{\cK}{\mathcal{K}}

\newcommand{\sC}{\mathscr{C}}
\newcommand{\sD}{\mathscr{D}} 

\newcommand{\sR}{\mathscr{R}}

\newcommand{\sW}{\mathscr{W}}

\newcommand{\fS}{\mathfrak{S}}

\renewcommand{\Im}{\mathrm{Im}\,}
\renewcommand{\Re}{\mathrm{Re}\,}
\begin{document}
\title[Ground states in relativistic QED]{
Ground states of semi-relativistic Pauli-Fierz and
no-pair Hamiltonians in QED at critical Coulomb coupling}
\author[Martin K\"onenberg {\protect \and} Oliver Matte]{Martin K\"onenberg {\protect
\and} Oliver Matte}
\begin{abstract}
We consider the semi-relativistic Pauli-Fierz Hamiltonian
and a no-pair model of a hydrogen-like atom interacting
with a quantized photon field at the respective critical
values of the Coulomb coupling constant.
For arbitrary values of the fine-structure constant and the
ultra-violet cutoff, we prove the existence of
normalizable ground states of the atomic system in both models.
This complements earlier results on the existence of
ground states in (semi-)relativistic models of
quantum electrodynamics at sub-critical coupling
by E.~Stockmeyer and the present authors.
Technically, the main new achievement is an improved estimate
on the spatial exponential localization of low-lying spectral
subspaces yielding uniform bounds at large Coulomb coupling constants.
In the semi-relativistic Pauli-Fierz model 
our exponential decay rate given in terms of the binding energy
reduces to the one known from the electronic model 
when the radiation field is turned off.
In particular, an increase of the binding energy
due to the radiation field is shown to improve the localization of ground
states.
\end{abstract}
\maketitle


\section{INTRODUCTION}

\noindent
By now the standard model
non-relativistic quantum electrodynamics (QED)
has been studied mathematically in great detail.
In this model non-relativistic electrons described
by molecular Schr\"odinger operators interact with a relativistic
quantized photon field via minimal coupling.
The resulting Hamiltonian is called the
non-relativistic Pauli-Fierz (NRPF) operator.
One may ask whether mathematical results on the NRPF operator
can be extended to models accounting for the electrons  
by relativistic operators as well.
There exist two such models whose
mathematical analysis seems canonical and interesting as an
intermediate step towards full QED,
where, besides the photon field,
also electrons and positrons are described as quantized fields.

The first model is given by the 
{\em semi}-relativistic Pauli-Fierz (SRPF) operator
where the non-relativistic kinetic energy of an electron 
in the NRPF model 
is replaced by its
square root.
This model has been treated mathematically for the
first time in \cite{MiyaoSpohn2009}. As mentioned by the authors
in that paper the SRPF model can be formally derived by applying 
a canonical quantization procedure
(described, e.g., in \cite{Spohn})
to a coupled system of differential equations comprised of
the Maxwell equations with a rigid charge distribution
and the evolution equation for a classical
relativistic point particle. By
choosing a static charge distribution
-- in our case the point charge of a nucleus
located at the origin -- one actually fixes a certain reference frame,
which is one
reason why the model is called semi-relativistic.
As they are interested in electron dynamics the authors
of \cite{MiyaoSpohn2009} include spin. A scalar square-root
Hamiltonian minimally coupled to the quantized radiation 
field appeared earlier in the mathematical analysis
of Rayleigh scattering \cite{FGS2001} which is related to the
relaxation of exited states to an atomic ground state.
In this situation the finite propagation 
speed of the electron is an advantageous feature of the dynamics
generated by square root Hamiltonains.

The second model treated here is a no-pair model introduced in
\cite{LiebLoss2002} in order to study the stability
of relativistic matter interacting with the
quantized radiation field. In this model
the Schr\"odinger operators in the NRPF Hamiltonian
are substituted by Dirac operators
and the whole Hamiltonian is restricted to a subspace
where all electrons live in positive spectral
subspaces of free Dirac operators with minimally coupled vector potentials.  
The idea to employ no-pair models to describe
atomic or molecular systems goes back to
\cite{BrwonRavenhall1951} and \cite{Sucher1980}.
In the latter paper a no-pair Hamiltonian 
is formally derived starting from full QED
by means of a procedure which neglects the
creation and annihilation of electron-positron pairs
-- which explains the nomenclature.
Nowadays, various no-pair models are extensively used, for instance, 
to include relativistic corrections in numerical
computations in quantum chemistry; see, e.g., \cite{ReiherWolf}.
There is always a certain freedom in choosing the
spectral subspaces determining a no-pair model.
The conventions in \cite{BrwonRavenhall1951,Sucher1980}, for instance,
force the electrons to live in positive
spectral subspaces of a Dirac operator {\em without}
magnetic vector potential.
Investigations of the stability of relativistic matter revealed,
however, that this choice --besides breaking gauge invariance --
always leads to instability
as soon as more than one electron is considered and the
interaction with the ever present
(classical or quantized) radiation field 
generated by the electrons is taken into account 
\cite{GriesemerTix1999,LiebLoss2002,LSS1997}.

In the case of a hydrogen-like atom -- that is, for one electron --
both models mentioned above are introduced in detail
in Section~\ref{sec-ops} after some notation has been
fixed in Section~\ref{sec-nora}.
As already indicated they have both been investigated in the mathematical
literature before 
\cite{FGS2001,HiroshimaSasaki2010,LiebLoss2002b,LiebLoss2002,MiyaoSpohn2009}, 
but to a much lesser extend than 
models of non-relativistic QED.
Their mathematical analysis is actually more difficult
than in the non-relativistic case since 
the electronic and photonic degrees of freedom
are coupled by {\em non-local} operators, namely the square roots
and spectral projections of the Dirac operators, respectively.
In our earlier works 
together with E.~Stockmeyer \cite{KMS2009a,KMS2009b,MatteStockmeyer2009a}
we gave some further contributions to these models
by proving the existence of energy minimizing, exponentially localized
ground states of the atomic system --
a question which has been solved
in non-relativistic QED in 
\cite{BFS1998b,BFS1999,Griesemer2004,GLL2001,LiebLoss2003}.

Typically, in relativistic atomic models
there exist critical values, $\gc$, of the Coulomb coupling
constant, $\gamma\grg0$, restricting the range
where physically distinguished self-adjoint realizations
of the Hamiltonian can be found.
This is due to the fact that in relativistic Coulomb systems
both the (positive) kinetic and the (negative) potential energy scale
as one over the length.
(In the physical application we have $\gamma=e^2Z$,
where $e^2$ is the square of the elementary charge and $Z\grg0$
is the atomic number.)
For the SRPF operator the critical value is equal to
the critical constant in Kato's inequality, $2/\pi$.
In the no-pair model the critical value is the one of
the (purely electronic) Brown-Ravenhall operator,
$2/(2/\pi+\pi/2)$ \cite{EPS1996}. 
According to \cite{KMS2009a,KMS2009b} these critical  
values do not change when the interaction with the quantized
photon field is taken into account.
The main results of \cite{KMS2009a,KMS2009b} hold,
however, only for sub-critical $\gamma$.
In particular, the existence of ground states of hydrogen-like atoms
at critical Coulomb coupling in the SRPF and no-pair models
has not yet been proven
and we wish to close this gap in the present article. 
We think that the analysis of the critical case is
interesting for several reasons.
First, the existence of an energy minimizing ground state
of an atomic Hamiltonian which is kept fixed
under the time evolution is a very fundamental notion in quantum theory. 
It is hence desirable to show that this conception holds true in our situation
as soon as the definition of the Hamiltonian makes sense.
Second, the mathematical investigation of the critical case
is interesting in its own right as one cannot employ simple
relative form bounds in order to control the Coulomb potential
and has to make use of more refined estimates instead.
Finally, by avoiding arguments requiring
the Coulomb potential to be a small form perturbation
we shall obtain estimates on the spatial exponential localization
of low-energy states which are uniform in the Coulomb
coupling constant and show that the decay rate of ground state eigenvectors
increases as a function of the binding energy, as $\gamma$ grows and 
approaches its critical value.
This improves on earlier results on the localization
of low-lying spectral subspaces of the SRPF and no-pair
operators which provide only $\gamma$-dependent
estimates \cite{MatteStockmeyer2009a}.
In particular, we obtain a more realistic description
of localization in our models, also for subcritical $\gamma$.
In the SRPF model our new localization estimates can
also be used to show that an increase of the binding energy
due to the quantized radition field (at fixed $\gamma$)
leads to an improved localization 
of low-energy states. 
(In fact, by simply ignoring the radiation field we can extend
well-known exponential localization estimates in $L^2$ for the 
purely electronic square-root operator
to all values of $\gamma\in[0,\nf{2}{\pi}]$. In the literature
we only found localization results for $\gamma\in(0,\nf{1}{2}]$
\cite{Nardini1986}.)

Presumably it is possible to
directly prove the existence of ground
states along the lines of 
\cite{BFS1998b,GLL2001,KMS2009a,KMS2009b},
also for $\gamma=\gc$.
We think, however, 
that it would be quite a tedious procedure
to replace {\em all} arguments in 
\cite{KMS2009a,KMS2009b} that exploit the sub-criticality of $\gamma$
by alternative ones.
For instance, simple characterizations of
the form domains of the Hamiltonians are available, for sub-critical
$\gamma$, which is very convenient in order to
argue that certain formal computations can be justified
rigorously. 
Therefore, it seems more comfortable
to pick some family of ground state eigenvectors, 
$\{\phi_\gamma\}_{\gamma<\gc}$, and
consider the limit $\gamma\nearrow\gc$.
To this end we shall apply a compactness argument in Section~\ref{sec-cpt}
similar to one used in \cite{GLL2001} in order
to remove an artificial photon mass.
Among other ingredients this compactness argument
requires the above-mentioned bound on the spatial localization
of $\phi_\gamma$, which is uniform for $\gamma\nearrow\gc$.
More precisely, we shall prove a suitable
bound on the spatial exponential localization
of spectral subspaces corresponding to energies below the
ionization threshold
which applies to all $\gamma\klg\gc$.
This localization estimate is derived in
Section~\ref{sec-exp-loc} by adapting and extending
ideas from \cite{BFS1998b,MatteStockmeyer2008b,MatteStockmeyer2009a}.
We remark that by now we are able to improve
the localization estimates of \cite{MatteStockmeyer2009a} thanks
to some more recent results of \cite{KMS2009b} collected
in Proposition~\ref{prop-iris}.
At the end of Section~\ref{sec-exp-loc} we discuss the
electron Hamiltonians without radiation field and the
improvement of localization in the SRPF model
when the coupling to the radiation field is turned on.
An important prerequisite for the analysis of both non-local
models studied here are commutator estimates involving
sign functions of Dirac operators, multiplication operators,
and the radiation field energy. Many such estimates
have been derived in 
\cite{KMS2009a,KMS2009b,MatteStockmeyer2008b,MatteStockmeyer2009a}.
For our new proof of the exponential localization we need,
however, still some additional ones. For this reason, and also
to make this paper self-contained and the proofs comprehensible,
we derive all required commutator estimates in Appendix~\ref{app-comm}.

The main results of this paper are Theorem~\ref{thm-exp-loc}
(Exponential localization) and Theorem~\ref{thm-ex-crit}
(Existence of ground states at critical coupling).

\section{NOTATION}\label{sec-nora}

\noindent
The Hilbert space underlying the atomic models studied
in this article is 
\begin{equation}\label{def-HR}
\HR:=\,L^2(\RR^3_\V{x},\CC^4)\otimes\Fock[\HP]
=\int_{\RR^3}^\oplus
\CC^4\otimes\Fock[\HP]\,d^3\V{x}\,,
\end{equation}
or a certain subspace of it.
Here 
$\Fock[\HP]=\bigoplus_{n=0}^\infty\Fock^{(n)}[\HP]$
denotes the bosonic Fock space
modeled over the one photon Hilbert space 
\begin{equation*}
\HP:=L^2(\RR^3\times\ZZ_2,dk)\,,\quad
\int dk:=\sum_{\lambda\in\ZZ_2}\int_{\RR^3}\!\!d^3\V{k}\,.
\end{equation*}
The letter
$k=(\V{k},\lambda)$ always
denotes a tuple consisting of a photon wave vector,
$\V{k}\in\RR^3$, and a polarization label, $\lambda\in\ZZ_2$.
The components of $\V{k}$ are denoted as
$\V{k}=(k^{(1)},k^{(2)},k^{(3)})$. We recall that
$\Fock^{(0)}[\HP]:=\CC$ and, for $n\in\NN$, 
$\Fock^{(n)}[\HP]:=\cS_n L^2((\RR^3\times\ZZ_2)^n)$,
where, for $\psi^{(n)}\in L^2((\RR^3\times\ZZ_2)^n)$,
$$
(\cS_n\,\psi^{(n)})(k_1,\ldots,k_n):=
\frac{1}{n!}\sum_{\pi\in\fS_n}\psi^{(n)}(k_{\pi(1)},\ldots,k_{\pi(n)})\,,
$$
$\fS_n$ denoting the group of permutations of $\{1,\ldots,n\}$.
For $f\in\HP$ and $n\in\NN_0$, we further define 
$\ad(f)^{(n)}:\Fock^{(n)}[\HP]\to\Fock^{(n+1)}[\HP]$
by $\ad(f)^{(n)}\,\psi^{(n)}:=\sqrt{n+1}\,\cS_{n+1}(f\otimes\psi^{(n)})$.
Then the standard bosonic creation operator is given by
$\ad(f)\,\psi:=
(0,\ad(f)^{(0)}\,\psi^{(0)},\ad(f)^{(1)}\,\psi^{(1)},\ldots\;)$,
for all $\psi=(\psi^{(0)},\psi^{(1)},\ldots\;)\in\Fock[\HP_m]$
such that the right side again belongs to $\Fock[\HP_m]$.
The corresponding annihilation operator is defined by 
$a(f):=\ad(f)^*$ and we have
the canonical commutation relations
\begin{equation}\label{CCR}
[a^\sharp(f)\,,\,a^\sharp(g)]=0\,,\qquad
[a(f)\,,\,\ad(g)]=\SPn{f}{g}\,\id\,,\qquad
f,g\in\HP,
\end{equation}
where $a^\sharp$ is $\ad$ or $a$.
Writing
\begin{equation}\label{def-kbot}
\V{k}_\bot\,:=\,(k^{(2)}\,,\,-k^{(1)}\,,\,0)\,,\qquad
\V{k}=(k^{(1)},k^{(2)},k^{(3)})\in\RR^3,
\end{equation}
we introduce two polarization vectors,
\begin{equation}\label{pol-vec}
\veps(\V{k},0)\,=\,
\frac{\V{k}_\bot}{|\V{k}_\bot|}
\,,\qquad
\veps(\V{k},1)\,=\,
\frac{\V{k}}{|\V{k}|}\,\wedge\,\veps(\V{k},0)\,,
\end{equation}
for almost every $\V{k}\in\RR^3$. Moreover, we introduce a
coupling function,
\begin{equation}\label{def-Gphys}
\V{G}_\V{x}(k)=\big(G_\V{x}^{(1)},G_\V{x}^{(2)},G_\V{x}^{(3)}\big)(k)
:=-e\,\frac{\id_{\{|\V{k}|\klg\UV\}}}{2\pi{|\V{k}|^{\nf{1}{2}}}}
\,e^{-i\V{k}\cdot\V{x}}\,\veps(k)\,,
\end{equation}
for all $\V{x}\in\RR^3$
and almost every $k=(\V{k},\lambda)\in\RR^3\times\ZZ_2$.
The values of the ultra-violet cut-off, $\UV>0$, and 
$e\in\RR$ are arbitrary.
(In the physical application $e$ is the square root
of Sommerfeld's fine structure constant and $e^2\approx1/137$.)
For short, we write
$a^\sharp(\V{G}_\V{x})
:=(a^\sharp(G_\V{x}^{(1)}),a^\sharp(G_\V{x}^{(2)}),a^\sharp(G_\V{x}^{(3)}))$.
Then the quantized vector potential
is the triple of operators
$\V{A}=(A^{(1)},A^{(2)},A^{(3)})$ in $\HR$ given as
\begin{equation}\label{def-Aphys}
\V{A}:=\int^\oplus_{\RR^3}
\id_{\CC^4}\otimes\big(\ad(\V{G}_\V{x})+a(\V{G}_\V{x})\big)\,d^3\V{x}\,.
\end{equation}
The radiation field energy is the
second quantization, $\Hf:=d\Gamma(\omega)$,
of the dispersion relation $\omega:\RR^3\times\ZZ_2\to\RR$,
$k=(\V{k},\lambda)\mapsto\omega(k):=|\V{k}|$.
By definition, $d\Gamma(\omega)$
is the direct sum
$d\Gamma(\omega):=\bigoplus_{n=0}^\infty d\Gamma^{(n)}(\omega)$,
where $d\Gamma^{(0)}(\omega):=0$, and
$d\Gamma^{(n)}(\omega)$ is
the maximal multiplication operator in $\Fock^{(n)}[\HP]$
associated with the symmetric function 
$(k_1,\ldots,k_n)\mapsto\omega(k_1)+\dots+\omega(k_n)$, if $n\in\NN$.

As usual we shall consider operators in
$L^2(\RR^3_\V{x},\CC^4)$ or $\Fock[\HP]$ also
as operators acting in the tensor product $\HR$
by identifying $|\hat{\V{x}}|^{-1}\equiv|\hat{\V{x}}|^{-1}\otimes\id$,
$\Hf\equiv\id\otimes\Hf$, etc.
(The hat $\hat{{}}$ indicates multiplication operators.)

Next, let $\alpha_0,\alpha_1,\alpha_2,\alpha_3$
denote hermitian 4\texttimes4 Dirac matrices
obeying the Clifford algebra relations
 \begin{equation}\label{Clifford}
\alpha_i\,\alpha_j+\alpha_j\,\alpha_i\,=\,2\,\delta_{ij}\,\id\,,
\qquad i,j\in\{0,1,2,3\}\,.
\end{equation}
In what follows they act on the second tensor factor in
$\HR=L^2(\RR^3_\V{x})\otimes\CC^4\otimes\Fock[\HP]$.
Then the free Dirac operator minimally coupled to $\V{A}$ is given as
\begin{equation}\label{def-DA}
\DA:=\valpha\cdot(-i\nabla_\V{x}+\V{A})+\alpha_0
:=
\sum_{j=1}^3\alpha_j\,(-i\partial_{x_j}+A^{(j)})+\alpha_0
\,.
\end{equation}
It is clear that
$\DA$ is well-defined a priori on the dense domain
$$
\core:=C_0^\infty(\RR^3,\CC^4)\otimes\sC\,,
\quad\textrm{(algebraic tensor product)}
$$
where $\sC\subset\Fock[\HP]$ denotes the subspace of
all $(\psi^{(n)})_{n=0}^\infty\in\Fock[\HP]$ such that
only finitely many components $\psi^{(n)}$ are non-zero
and such that each $\psi^{(n)}$, $n\in\NN$, is essentially
bounded with compact support.
Moreover, it is well-known that
$\DA$ is essentially self-adjoint on $\sD$; see, e.g., \cite{LiebLoss2002}.
We use the symbol
$\DA$ again to denote its closure starting from $\sD$.

Finally, we use the symbols $\dom(T)$ and $\form(T)$
to denote the domain and form domain, respectively,
of some suitable operator $T$. We further put $a\wedge b:=\min\{a,b\}$, 
$a\vee b:=\max\{a,b\}$, $a,b\in\RR$,
and $\SL y\SR:=(1+y^2)^{\nf{1}{2}}$,
$y\in\RR$. The symbols $C(a,b,\ldots\,),C'(a,b,\ldots\,),\ldots\,$
denote positive constants which depend only on the
quantities $a,b,\ldots$ displayed in their arguments
and whose values might change from one estimate to another.

\section{THE SEMI-RELATIVISTIC PAULI-FIERZ AND 
NO-PAIR MODELS}\label{sec-ops}

\noindent
In what follows we shall denote the maximal operator of
multiplication with the Coulomb potential, $-\gamma/|\V{x}|$, $\gamma\grg0$,
in $\HR$ by $V_\gamma$.
Then the semi-relativistic Pauli-Fierz (SRPF) operator
is defined, a priori on the dense domain $\core$, as
$$
\PF{\gamma}:=|\DA|+V_\gamma+\Hf\,.
$$
Notice that the absolute value $|\DA|$ is actually
a square root operator minimally coupled to $\V{A}$.
For, if the Dirac matrices are given in the standard representation,
then 
$$
|\DA|=\cT_\V{A}^{\nf{1}{2}}\oplus\cT_\V{A}^{\nf{1}{2}},\qquad
\cT_\V{A}:=(\vsigma\cdot(-i\nabla_\V{x}+\V{A}))^2+\id\,,
$$
where $\vsigma$ is a formal vector containing the three 2\texttimes2
Pauli spin matrices.
According to \cite{KMS2009b} the quadratic form associated
with $\PF{\gamma}$ is semi-bounded below, if and only if
$\gamma$ is less than or equal to the critical constant
in Kato's inequality,
$$
\gcPF:=2/\pi\,.
$$
Thus, the range of stability of $\PF{\gamma}$ is the same
as the one of the purely electronic square root, or,
Herbst \cite{Herbst1977} operator,
\begin{equation}\label{chandra}
H_\gamma^{\el,\Pf}:=\sqrt{1-\Delta_\V{x}}+V_\gamma\,.
\end{equation}
From now on the symbol $\PF{\gamma}$ will again denote
the Friedrichs extension of the SRPF operator, provided that
$\gamma\in[0,\gcPF]$.

Compared to the non-relativistic Pauli-Fierz model
there are only a few mathematical works dealing with
its semi-relativistic analogue:
Spinless square root operators coupled to quantized fields
appear in the study of Rayleigh scattering in \cite{FGS2001}
and the fiber decomposition of $\PF{\gamma=0}$ is 
investigated in \cite{MiyaoSpohn2009}. To recall some further results
we define the ionization threshold and the ground state energy
of $\PF{\gamma}$, respectively, as
$$
\ThPF:=\inf\spec[\PF{0}]\,,\qquad
E_\gamma^\Pf:=\inf\spec[\PF{\gamma}]\,,\;\;\gamma\in(0,\gcPF]\,.
$$
Then the following  
shall be relevant for us:

\begin{proposition}[\cite{HiroshimaSasaki2010,KMS2009a}]\label{prop-PF-sbc}
(i) For all $e\in\RR$, $\UV>0$, and $\gamma\in(0,\gcPF]$,
\begin{equation}\label{IE-PF}
\ThPF-E_\gamma^\Pf\grg
1-\inf\spec[H_\gamma^{\el,\Pf}]>0\,.
\end{equation}
(ii) For all $e\in\RR$, $\UV>0$, and $\gamma\in(0,\gcPF)$,
$E_\gamma^\Pf$ is an eigenvalue of $\PF{\gamma}$.
\end{proposition}

\begin{proof}
Part~(i) follows from \cite{HiroshimaSasaki2010} (at least
in the case $\gamma\in(0,\nf{1}{2})$, where $\PF{\gamma}$ 
is essentially self-adjoint on $\core$ \cite{KMS2009b}).
An alternative proof of \eqref{IE-PF} covering all 
$\gamma\in(0,\gcPF]$ can be found
in \cite{KMS2009a}.
Part~(ii) is the main result of \cite{KMS2009a}.
\end{proof}

In the present paper we shall extend
the results on the spatial exponential localization 
of spectral subspaces below $\ThPF$ of $\PF{\gamma}$, $\gamma\in(0,\gcPF)$,
\cite{MatteStockmeyer2009a} and Proposition~\ref{prop-PF-sbc}(ii)
to the critical case $\gamma=\gcPF$.

In order to introduce the second model studied in this paper
we first recall that
the spectrum of $\DA$ consists of two half-lines,
$
\spec(\DA)=(-\infty,-1]\cup[1,\infty)
$.
We denote the orthogonal projections
onto the positive and negative spectral subspaces by
\begin{equation*}
\PApm:=\id_{\RR^\pm}(\DA)=
\frac{1}{2}\,\id\pm\frac{1}{2}\,\SA\,,\qquad
\SA:=\DA\,|\DA|^{-1}.
\end{equation*}
Then the no-pair operator 
is a self-adjoint
operator acting in the positive spectral subspace $\PA\HR$
defined, a priori on the dense domain
$\PA\,\sD\subset\PA\HR$, by
\begin{equation}\label{def-NPhypG}
\NP{\gamma}:=
\PA\,(\DA+V_\gamma+\Hf)\,\PA\,.
\end{equation}
Thanks to \cite[Proof of Lemma~3.4(ii)]{MatteStockmeyer2009a},
which implies that $\PA$ maps the subspace $\sD$ into 
$\dom(\DO)\cap\dom(\Hf^\nu)$, for every $\nu>0$,
and Hardy's inequality,
we actually know that $\NP{\gamma}$ is well-defined on $\sD$.
Due to \cite{KMS2009b}
the quadratic form associated with $\NP{\gamma}$
is semi-bounded below, if and only if $\gamma$ is less than or equal to
$$
\gcnp:=2/(2/\pi+\pi/2)\,,
$$
which is the critical constant for the stability
of the electronic Brown-Ravenhall operator,
\begin{equation}\label{charly}
H_\gamma^{\el,\np}:=\PO\,(\DO+V_\gamma)\,\PO\,.
\end{equation}
The value of $\gcnp$ has been determined in \cite{EPS1996}.
Again we denote the Friedrichs extension of the no-pair
operator by the same symbol $\NP{\gamma}$, if $\gamma\in[0,\gcnp]$.
Because of technical reasons it is 
convenient to add the following counter-part
acting in the negative spectral subspace $\PAm\,\HR$,
$$
\NPneg{\gamma}
:=\PAm\,(-\DA+V_\gamma+\Hf)\,\PAm\,,\quad\gamma\in [0,\gcnp]\,,
$$
which is also defined as a Friedrichs extension starting from $\core$.
In fact, $\NP{\gamma}$ and $\NPneg{\gamma}$ are unitarily
equivalent as the unitary and symmetric
matrix $\vt:= \alpha_1\,\alpha_2\,\alpha_3\,\alpha_0$ 
anti-commutes with $\DA$, so that $\vt\,\PA=\PAm\,\vt$.
Thus, if questions like localization and existence of
ground states are addressed, then we may equally
well consider the operator
\begin{equation}\label{NPFNP}
\FNP{\gamma}:=
\NP{\gamma}\oplus\NPneg{\gamma}=\NP{\gamma}\oplus\{\vt\,\NP{\gamma}\vt\}\,.
\end{equation}
For later reference we observe that
\begin{equation}\label{def-FNP}
\FNP{\gamma}
=|\DA|+\frac{1}{2}\,(V_\gamma+\Hf)
+\frac{1}{2}\,\SA\,(V_\gamma+\Hf)\,\SA\quad
\textrm{on}\;\core.
\end{equation}
The mathematical analysis of a molecular analogue of
$\NP{\gamma}$ has been initiated in \cite{LiebLoss2002}
where the stability of the second kind of relativistic matter 
has been established in the no-pair model under certain
restrictions on $e$, $\UV$, and the nuclear charges.
Moreover, an upper bound on the (positive)
binding energy is derived in \cite{LiebLoss2002b}.
To recall some results on hydrogen-like atoms
used later on we put
$$
\Thnp:=\inf\spec[\FNP{0}]\,,\qquad
E_\gamma^\np:=\inf\spec[\FNP{\gamma}]\,,\;\;\gamma\in(0,\gcnp]\,.
$$
Both parts of the following proposition are proven
in \cite{KMS2009b}:

\begin{proposition}[\cite{KMS2009b}]\label{prop-np-sbc}
(i) For all $e\in\RR$, $\UV>0$, and $\gamma\in(0,\gcnp]$,
there is some $c(\gamma,e,\UV)>0$ such that
\begin{equation}\label{IE-np}
\Thnp-E_\gamma^\np\grg c(\gamma,e,\UV)\,.
\end{equation}
(ii) For all $e\in\RR$, $\UV>0$, and $\gamma\in(0,\gcnp)$,
$E_\gamma^\np$ is an eigenvalue of $\FNP{\gamma}$.
\end{proposition}

The exponential localization
of spectral subspaces corresponding to energies below
$\Thnp$ is shown in \cite{MatteStockmeyer2009a},
again for sub-critical values
of $\gamma$ only. We propose to extend the latter
result as well as Proposition~\ref{prop-np-sbc}(ii)
to the case $\gamma=\gcnp$ in the present article.

We close this section by
recalling some further results of \cite{KMS2009b}
used later on.
In order to improve the localization estimates
of \cite{MatteStockmeyer2009a} and to deal with
critical coupling constants
the bounds in \eqref{iris} below are particularly important.
For they allow to control small pieces of the electronic kinetic energy
by the total Hamiltonian even in the critical cases.
Their proofs involve a strengthened version
of the sharp generalized Hardy inequality
obtained recently in
\cite{Frank2009,SoSoeSp2008} and an analogous
inequality for the Brown-Ravenhall model \cite{Frank2009}.

\begin{proposition}[\cite{KMS2009b}]\label{prop-iris}
Let $\gc$ be $\gcPF$ or $\gcnp$ and
$H_\gamma$ be $\PF{\gamma}$ or $\FNP{\gamma}$.
Then, for all $e\in\RR$ and $\UV>0$, the following holds:

\smallskip

\noindent
(i) For $\gamma\in[0,\nf{1}{2})$, $H_\gamma$ is essentially
self-adjoint on $\core$.

\smallskip

\noindent
(ii) For all $\ve\in(0,1)$, $\delta>0$, and $\gamma\in[0,\gc]$, 
\begin{align}\label{iris}
|\DO|^{\ve}
&\klg\delta\,H_{\gamma}+C(e,\UV,\delta,\ve)\,,\qquad
|\DA|^{\ve}\klg\delta\,H_{\gamma}+C'(e,\UV,\delta,\ve)\,,
\end{align}
in the sense of quadratic forms on $\form(H_\gamma)$.
\smallskip

\noindent
(iii) $\dom(H_\gamma)\subset\dom(\Hf)$ and,
for all $\delta>0$, $\gamma\in[0,\gc]$, and $\psi\in\dom(H_\gamma)$,
\begin{equation}\label{iris5}
\|\Hf\,\psi\|\klg(1+\delta)\,\|H_\gamma\,\psi\|+C(e,\UV,\delta)\,\|\psi\|\,.
\end{equation}
\end{proposition}


\section{EXPONENTIAL LOCALIZATION}\label{sec-exp-loc}

\noindent
In this section we show that low-lying spectral subspaces of
$\PF{\gamma}$ and $\FNP{\gamma}$ are exponentially localized
with respect to $\V{x}$
in a $L^2$ sense. This result is stated and proven in
Theorem~\ref{thm-exp-loc} later on.
The general idea behind its proof,
which rests on a simple identity involving the 
spectral projection (see \eqref{BFS})
and the Helffer-Sj\"ostrand formula, 
is due to \cite{BFS1998b}. (More precisely, \eqref{BFS}
is variant of a similar identity used in \cite{BFS1998b}. It
has been employed earlier in \cite{MatteStockmeyer2008b}.)
From a technical point of view
the key step in the proof consists, however, in
showing that the resolvent of a certain comparison operator
stays bounded after the conjugation with exponential
weights (Lemma~\ref{le-HR}). Moreover, one has to derive
a useful resolvent identity involving the comparison operator
and the original one (Lemma~\ref{le-res-id1}).
In these steps our arguments are more streamlined and
simpler than those used in the earlier paper \cite{MatteStockmeyer2009a}
as we work with a simpler comparison operator. Moreover,
we now treat critical $\gamma$ as well.
By now these improvements are possible thanks to 
the results of \cite{KMS2009b} collected in
Proposition~\ref{prop-iris}.

In the whole section we fix some $\mu\in C_0^\infty(\RR^3_\V{x},[0,1])$
such that $\mu=1$ on $\{|\V{x}|\klg1\}$ and $\mu=0$
on $\{|\V{x}|\grg2\}$ and set $\mu_R(\V{x}):=\mu(\V{x}/R)$,
for all $\V{x}\in\RR^3$ and $R\grg1$.
Then we put
\begin{align*}
V_{\gamma,R}&:=(1-\mu_R)\,V_\gamma
=(\mu_R-1)\,\gamma/|\hat{\V{x}}|\,,
\end{align*}
and define two comparison operators (compare \eqref{def-FNP}),
\begin{align*}
\PF{\gamma,R}&:=|\DA|+V_{\gamma,R}+\Hf\,,\quad\gamma\in(0,\gcPF]\,,
\\
\FNP{\gamma,R}&:=|\DA|+\frac{1}{2}\,(V_{\gamma,R}+\Hf)
+\frac{1}{2}\,\SA\,(V_{\gamma,R}+\Hf)\,\SA\,,
\quad\gamma\in(0,\gcnp]\,,
\end{align*}
on the domain $\core$ to start with.
According to Proposition~\ref{prop-iris}(i)
both operators then are essentially self-adjoint
and we again use the symbols $\PF{\gamma,R}$ and $\FNP{\gamma,R}$
to denote their self-adjoint closures.
Clearly,
\begin{equation}\label{nora00}
\PF{\gamma,R}\grg\ThPF-\|V_{\gamma,R}\|_\infty\,,\qquad
\FNP{\gamma,R}\grg\Thnp-\|V_{\gamma,R}\|_\infty\,,
\end{equation}
where $\|V_{\gamma,R}\|_\infty\klg1/R$, $R\grg1$.
In order to treat both models at the same time 
we shall use the following notation from now on:
\begin{equation}
\label{nora}\left\{
\begin{array}{l}
\textrm{The symbols}\;\,H,\,H_R,\,\Th,\,E\;\,
\textrm{denote either}\\
\PF{\gamma},\,\PF{\gamma,R},\,\Th^\Pf,\,E^\Pf_\gamma\;\,
\textrm{or}\;\,\FNP{\gamma},\,\FNP{\gamma,R},\,\Th^\np,\,E^\np_\gamma\,.
\end{array}
\right.
\end{equation}
Since the domains of $H$ and $H_R$ will in general be different
we cannot compare their resolvents by means of the second
resolvent identity. To overcome this problem we shall regularize
the difference of their resolvents by means of the following
cut-off function, which is also kept fixed throughout the whole section:

We pick some $\chi\in C^\infty(\RR^3_\V{x},[0,1])$ such that
$\chi=0$ on $\{|\V{x}|\klg2\}$ and $\chi=1$
on $\{|\V{x}|\grg4\}$ and set
$\chi_R(\V{x}):=\chi(\V{x}/R)$, for all $\V{x}\in\RR^3$ and $R\grg1$.

Finally, we introduce a class of weight functions,
\begin{equation*}
\sW_{a}:=\big\{
F\in C^\infty(\RR^3_{\V{x}},[0,\infty)):\;F(\V{0})=0\,,\;
\|F\|_\infty<\infty\,,\;|\nabla F|\klg a\big\},
\end{equation*}
where $a\in(0,1)$,
and define two families of operators on the dense domain $\core$,
\begin{align*}
U_R^F(z)&:=(H-z)^{-1}\,(H_R-H)\,\chi_R\,e^F,
\\
W_R^F(z)&:=(H-z)^{-1}\,[\chi_R\,,\,H_R]\,e^F,
\end{align*}
for $z\in\CC\setminus\RR$, $R\grg1$, $F\in\sW_a$, and $a\in(0,1)$.
Since $(V_\gamma-V_{\gamma,R})\,\chi_R=0$ we actually have $U_R^F(z)=0$
when $H=\PF{\gamma}$.

In the whole section the positive
constants $C(a,e,\ldots),C'(a,\ldots),\ldots\,$ 
are increasing
functions of each displayed parameter when the others are kept fixed.

\begin{lemma}\label{le-UV}
Let $z\in\CC\setminus\RR$, $R\grg1$, $F\in\sW_a$, and $a\in(0,1)$.
Then $U_R^F(z)$ and $W_R^F(z)$ extend to bounded operators on
$\HR$ and
\begin{align*}
\sup_{F\in\sW_a}\|U_R^F(z)\|+
\sup_{F\in\sW_a}\|W_R^F(z)\|\klg
C(a,e,\UV,R)\,\frac{1\vee|\Re z|}{1\wedge|\Im z|}.
\end{align*}
\end{lemma}

\begin{proof}
In the case of the no-pair operator we have
$$
U_R^F(z)=\frac{1}{2}\,\SA\,(\FNP{\gamma}-z)^{-1}\,(V_{\gamma,R}-V_\gamma)\,e^F\,
\big[e^{-F}\,\SA\,e^F,\,\chi_R\big]\quad \textrm{on}\;\core\,,
$$
where we used $[\FNP{\gamma},\SA]=0=(V_\gamma-V_{\gamma,R})\,\chi_R$.
In Lemma~\ref{le-sandra1} we shall show that
$$
\big\|\,|\hat{\V{x}}|^{-\kappa}\,(\Hf+1)^{-\nicefrac{1}{2}}\,
[e^{-F}\,\SA\,e^F,\,\chi_R]\,\big\|\klg C(a,e,\UV,\kappa)\,\|\nabla\chi\|/R\,,
$$
for every $\kappa\in[0,1)$. 
Combining the previous bound with
the following consequence of 
$|\hat{\V{x}}|^{-\nf{1}{2}}\klg C\,|\DO|^{\nf{1}{2}}$,
\eqref{iris}, and \eqref{iris5},
\begin{align*}
\big\|\,|\hat{\V{x}}|^{-\nf{1}{8}}\,\Hf^{\nicefrac{1}{2}}\,\psi\big\|^2
\klg\big\|\,|\hat{\V{x}}|^{-\nf{1}{4}}\,\psi\big\|\,\big\|\Hf\,\psi\big\|\klg
C(e,\UV)\,\frac{1\vee|\Re z|^2}{1\wedge|\Im z|^2}\,
\big\|(\FNP{\gamma}-\ol{z})\,\psi\big\|^2,
\end{align*}
for every $\psi\in\dom(\FNP{\gamma})$,
we deduce that
\begin{align*}
\big\|U_R^F(z)\,\vp\big\|&\klg\frac{1}{2}\,
\big\|\,|\hat{\V{x}}|^{-\nf{1}{8}}\,
(\Hf+1)^{\nicefrac{1}{2}}(\FNP{\gamma}-\ol{z})^{-1}\big\|\,\|e^F\mu_R\|
\\
&\qquad\qquad\cdot
\big\|\,|\hat{\V{x}}|^{-\nf{7}{8}}(\Hf+1)^{-\nicefrac{1}{2}}
[e^{-F}\,\SA\,e^F,\chi_R]\,\vp\big\|
\\
&\klg
C'(a,e,\UV)\,(e^{2aR}/R)\,
\frac{1\vee|\Re z|}{1\wedge|\Im z|}\,\|\vp\|\,,\qquad \vp\in\core\,.
\end{align*}
Next, we turn our attention to $W_R^F(z)$. In the case of the SRPF operator
we have $[\chi_R\,,\,\PF{\gamma,R}]=[\chi_R\,,\,|\DA|\,]$,
and it follows from Lemma~\ref{le-sandra2} that,
for all $F\in\sW_{a}$,
\begin{equation}\label{muddy0}
\big\|\,
[\chi_R,\SA]\,e^F\big\|+
\big\|\,|\DA|^{-\nf{1}{4}}
[\chi_R,|\DA|\,]\,e^F\big\|\klg C(a)\,\|\nabla\chi_R\,e^F\|_\infty
\klg C'(a,R)\,.
\end{equation}
Here we also used that $0\klg F\klg 4aR$ on $\supp(\nabla\chi_R)$.
On account of \eqref{iris} we also
have
$$
\big\|\,|\DA|^{\nf{1}{4}}(\PF{\gamma}-\ol{z})^{-1}\big\|\klg 
C(e,\UV)\,\frac{1\vee|\Re z|}{1\wedge|\Im z|}\,.
$$
Putting these remarks together
we arrive at the asserted bound on $W_R^F(z)$ for the
SRPF operator.

In the case of the no-pair operator
\begin{align}\nonumber
[\chi_R\,,\,\FNP{\gamma,R}]\,e^F
&=[\chi_R\,,\,|\DA|\,]\,e^F
\\\nonumber
&\quad+
\frac{1}{2}\,[\chi_R,\SA]\,e^F\,\Hf\,e^{-F}\SA\,e^F
+\frac{1}{2}\,\SA\,\Hf\,[\chi_R,\SA]\,e^F
\\\label{muddy1}
&\quad+\frac{1}{2}\,[\chi_R,\SA]\,e^F\,V_{\gamma,R}\,e^{-F}\SA\,e^F
+\frac{1}{2}\,\SA\,V_{\gamma,R}\,[\chi_R,\SA]\,e^F.
\end{align}
The first term on the RHS of \eqref{muddy1} is dealt with exactly
as in the case of the SRPF operator above.
Moreover, on account of \eqref{muddy0} 
and $\|e^{-F}\SA e^F\|\klg C(a)$ (see \eqref{clara2})
the norms of both operators in
the third line of \eqref{muddy1} are bounded by some 
constant depending only on $a$ and $R$.
By Lemma~\ref{le-sandra3} we finally have
$$
\big\|(\Hf+1)^{-1}\,[\chi_R\,,\,\SA]\,e^F\,\Hf\,\big\|
\klg C(a,e,\UV)\,\|\nabla\chi_R\,e^F\|_\infty\klg C(a,e,\UV)\,
\|\nabla\chi\|\,e^{4aR}/R\,,
$$
and we conclude by means of the following consequence of \eqref{iris5},
$$
\big\|\Hf\,\SA\,(\FNP{\gamma}-\ol{z})^{-1}\big\|=
\big\|\Hf\,(\FNP{\gamma}-\ol{z})^{-1}\big\|
\klg C(e,\UV)\,\frac{1\vee|\Re z|}{1\wedge|\Im z|}\,.
$$
Here we also use that $[\FNP{\gamma},\SA]=0$.
\end{proof}

\begin{lemma}\label{le-res-id1}
For all $z\in\CC\setminus\RR$, $R\grg1$, $F\in\sW_{a}$,
and $a\in(0,1)$,
\begin{equation*}
\chi_R\,(H-z)^{-1}-(H_R-z)^{-1}\chi_R=
(H_R-z)^{-1}e^{-F}\big(U_{R}^F(z)^*+W_{R}^F(z)^*\big)\,.
\end{equation*}
\end{lemma}

\begin{proof}
For all $\vp\in\core$, 
\begin{align*}
\big\{&(H-z)^{-1}\,\chi_R-\chi_R\,(H_R-z)^{-1}\big\}\,(H_R-z)\,\vp
\\
&=(H-z)^{-1}\,\chi_R\,(H_R-z)\,\vp-\chi_R\,\vp
\\
&=
(H-z)^{-1}\,(H_R-H+H-z)\,\chi_R\,\vp
-\chi_R\,\vp+(H-z)^{-1}\,[\chi_R\,,\,H_R]\,\vp
\\
&=\big\{(H-z)^{-1}\,(H_R-H)\,\chi_R\,e^F\big\}\,e^{-F}\vp
+\big\{(H-z)^{-1}\,[\chi_R\,,\,H_R]\,e^F\big\}\,e^{-F}\vp
\\
&=
\big(\ol{U}_{R}^F(z)+\ol{W}_{R}^F(z)\big)\,e^{-F}\,(H_R-z)^{-1}
(H_R-z)\,\vp\,.
\end{align*}
Now, $\ol{U}_{R}^F(z)$ and $\ol{W}_{R}^F(z)$ are bounded and 
$(H_R-z)\,\core$ is dense in $\HR$, as $H_R$ is essentially
self-adjoint on $\core$. Hence, we infer that
$$
(H-z)^{-1}\,\chi_R-\chi_R\,(H_R-z)^{-1}
=\big(\ol{U}_{R}^F(z)+\ol{W}_{R}^F(z)\big)\,
e^{-F}\,(H_R-z)^{-1}.
$$
Taking the adjoint of this operator identity and replacing
$\ol{z}$ by $z$ we arrive at the assertion.
\end{proof}

\begin{lemma}\label{le-retno}
For all $F\in\sW_{a}$, $a\in(0,1)$, and $\vp\in\core$,
\begin{equation}\label{retno1}
\Re\SPb{\vp}{e^F\,|\DA|\,e^{-F}\vp}\grg
\SPb{\vp}{(\DA^2-|\nabla F|^2)^{\nf{1}{2}}\,\vp}\,.
\end{equation}
\end{lemma}

\begin{proof}
For every non-negative operator, $T\grg0$, on some Hilbert space
and $\psi\in\dom(T)$, we have
$$
T^{\nf{1}{2}}\,\psi=
\int_0^{\infty}(T+\eta)^{-1}T\,\psi\,\frac{d\eta}{\pi\,\eta^{\nf{1}{2}}}
=\int_0^{\infty}\big(\id-\eta\,(T+\eta)^{-1}\big)\,\psi\,
\frac{d\eta}{\pi\,\eta^{\nf{1}{2}}}\,.
$$
For every $\vp\in\core$, this yields the formula
\begin{align*}
\SPb{\vp}{\big(e^F\,|\DA|\,e^{-F}-(\DA^2-|\nabla F|^2)^{\nf{1}{2}}\big)\,\vp}
=\int_0^{\infty}J[\vp;\eta]\,\frac{\eta^{\nf{1}{2}}d\eta}{\pi}\,,
\end{align*}
with 
\begin{align*}
J[\vp;\eta]&:=\SPb{\vp}{\big(\sR_F(\eta)-e^F\sR_0(\eta)\,e^{-F}\big)\,\vp}\,,
\\
\sR_G(\eta)&:=(\DA^2-|\nabla G|^2+\eta)^{-1},\quad G\in\{0,F\}\,.
\end{align*}
Now, let $\phi:=e^F(\DA^2+\eta)\,e^{-F}\psi$, for some $\psi\in\core$.
Then
\begin{align*}
\Re\SPb{\phi}{e^F\sR_0(\eta)\,e^{-F}\phi}
&=
\Re\SPb{e^F(\DA^2+\eta)\,e^{-F}\psi}{\psi}
\\
&=\SPb{(\DA^2-|\nabla F|^2+\eta)\,\psi}{\psi}
\grg(1-a^2+\eta)\,\|\psi\|^2\grg0\,.
\end{align*}
It is well-known that $\DA^2$ is essentially selfadjoint on
$\core$. Since, for every $F\in\sW_a$, multiplication with $e^{-F}$
maps $\core$ bijectively onto itself, this implies that
$(\DA^2+\eta)\,e^{-F}\core$ is dense in $\HR$.
Since $F\in\sW_a$ is bounded we conclude that the previous
estimates hold, for all $\phi$ in some dense domain,
whence $\Re[e^F\sR_0(\eta)\,e^{-F}]\grg0$ as a quadratic form
on $\HR$. Next, we set
$Q:=(\valpha\cdot\nabla F)\,\DA+\DA\,(\valpha\cdot\nabla F)$
and let 
$$
\vp:=(\DA^2-|\nabla F|^2+\eta)\,\psi
=e^{\pm F}(\DA^2+\eta)\,e^{\mp F}\psi\mp iQ\,\psi\,, 
$$
for $\psi\in\core$. Then
\begin{align*}
J[\vp;\eta]&=
i\SPb{e^{-F}\sR_0(\eta)\,e^{F}\vp}{Q\,\psi}
=i\SPn{\psi}{Q\,\psi}
+\SPb{Q\,\psi}{e^F\sR_0(\eta)\,e^{-F}Q\,\psi}\,.
\end{align*}
Here $\DA^2-|\nabla F|^2$ is essentially selfadjoint
on $\core$ and $Q$ is symmetric on the same domain.
Hence, $\Re J[\vp;\eta]\grg0$, for all $\vp$ in a dense set, 
thus for all $\vp\in\HR$, and we conclude.
\end{proof}

In what follows we set
\begin{equation}\label{def-robert}
\rho(a):=1-(1-a^2)^{\nf{1}{2}},\qquad a\in(0,1)\,.
\end{equation}

\begin{lemma}\label{le-HR}
For all $\delta>0$, $a\in(0,1)$, and $R\grg1$,
$$
\sup\big\{\,\|e^F(H_R-z)^{-1}e^{-F}\|\,:\;
F\in\sW_{a}\,,\;\Re z\klg\Sigma-\rho(a)
-g/R-h\,a^2-\delta\,\big\}
\klg\frac{1}{\delta}\,,
$$
where
$g=\gamma$ and $h=0$ in the case of the SRPF operator.
In the case of the no-pair operator
we may choose $g,h=C(a,e,\UV)$.
\end{lemma}

\begin{proof}
It suffices to show that,
for $\Re z\klg\Sigma-\rho(a)-g/R-h\,a^2-\delta$ 
and all $\psi\in\core$,
\begin{align}
\delta\,\|\psi\|^2&\klg
\Re\SPb{\psi}{e^F(H_R-z)\,e^{-F}\psi}
\klg\|\psi\|\,\big\|e^F(H_R-z)\,e^{-F}\psi\big\|\,.\label{sayuri1}
\end{align}
In fact, if $F\in\sW_{a}$, then $e^{-F}$ maps $\core$ 
bijectively onto itself, thus
$(H_R-z)\,e^{-F}\core$ is dense in $\HR$, as we know that
$H_R$ is essentially self-adjoint on $\core$ and $z\in\vr(H_R)$.
In particular, we may insert $\psi:=e^F(H_R-z)^{-1}e^{-F}\vp$,
$\vp\in\HR$, into \eqref{sayuri1}, since $F\in\sW_{a}$
is bounded, and this yields the assertion.

First, we prove \eqref{sayuri1} for the SRPF operator.
Since the square root is operator monotone, $|\nabla F|\klg a$,
and $|\DA|\grg1$, we have
$$
(\DA^2-|\nabla F|^2)^{\nf{1}{2}}\grg|\DA|+(\DA^2-a^2)^{\nf{1}{2}}-|\DA|
\grg|\DA|-\rho(a)\,.
$$
Applying \eqref{retno1} we deduce that,
as quadratic forms on $\core$,
\begin{align*}
\Re\big[e^F\,H_R\,e^{-F}\big]&=
\Re\big[e^F\,|\DA|\,e^{-F}\big]+V_{\gamma,R}+\Hf
\\
&\grg|\DA|+\Hf-\rho(a)+V_{\gamma,R}
\grg\Sigma^{\Pf}-\rho(a)-\gamma/R\,.
\end{align*}

In order to discuss the no-pair operator we put
$$
\SA^F:=e^F\,\SA\,e^{-F},
\quad\cK_F\,:=\,[e^F,\,\SA]\,e^{-F},\qquad \pm F\in\sW_{a}\,.
$$
According to \cite[Lemma~3.5]{MatteStockmeyer2009a}
(or Lemma~\ref{le-sandra2})
we have $\|\SA^F\|=\bigO(1)$ and $\|\cK_F\|=\bigO(a)$, as $a\searrow0$.
We further define
$$
\triangle(\Hf):=
e^F\,\SA \,\Hf\,\SA\,e^{-F}-\SA \,\Hf\,\SA
=\SA \,\Hf\,\cK_F+\cK_F\,\Hf\,\SA+\cK_F\,\Hf\,\cK_F\,.
$$
Then a brief computation using \eqref{def-FNP} gives
\begin{align*}
e^F\FNP{\gamma,R}e^{-F}&=
e^F|\DA|e^{-F}+\frac{1}{2}\big(V_{\gamma,R}+\Hf+\SA\,\Hf\,\SA\big)
+\frac{1}{2}\,\SA^FV_{\gamma,R}\SA^F
+\frac{1}{2}\triangle(\Hf)
\end{align*}
on $\core$, and similarly as above we obtain
\begin{align*}
\Re\big[e^F\,\FNP{\gamma,R}\,e^{-F}\big]
&\grg\FNP{0}-\rho(a)-\bigO(1)/R
+\Re\triangle(\Hf)/2
\end{align*}
on $\core$.
Furthermore,
\begin{align*}
&2\,|\SPn{\vp}{\Re\triangle(\Hf)\,\vp}|
\\
&\klg
a^2\,\big\|\Hf^{\nf{1}{2}}\,\SA\,\vp\big\|^2+\frac{1}{a^2}\,
\big\|\Hf^{\nf{1}{2}}(\cK_F+\cK_F^*)\,\vp\big\|^2
+\big\|\Hf^{\nf{1}{2}}\cK_F\,\vp\big\|^2
+\big\|\Hf^{\nf{1}{2}}\cK_F^*\,\vp\big\|^2
\\
&\klg a^2\,C(a,e,\UV)\,\SPn{\vp}{\FNP{0}\,\vp}\,,
\end{align*}
for all $\vp\in \sD$, where we used $\cK_F^*=\cK_{-F}$ and
\begin{equation}\label{clara33}
\big\|\Hf^{\nf{1}{2}}\cK_{\pm F}\,(\Hf+1)^{-\nf{1}{2}}\big\|^2+
\big\|\Hf^{\nf{1}{2}}(\cK_F+\cK_F^*)(\Hf+1)^{-\nf{1}{2}}\big\|\klg C'(a,e,\UV)\,a^2.
\end{equation}
in the second step. The bound \eqref{clara33} 
follows from \eqref{clara3} and \eqref{clara99} below.
In fact, $\cK_F+\cK_F^*$ is equal to the 
double commutator in \eqref{clara99}.
Therefore,
\begin{align*}
\Re\big[e^F\,\FNP{\gamma,R}\,e^{-F}\big]
&\grg(1-\bigO(a^2))
\,\FNP{0}-\rho(a)-\bigO(1)/R
\\
&\grg\Thnp-\rho(a)-\bigO(1)/R-\bigO(a^2)\,\Thnp\,,
\quad a\searrow0\,.\qedhere
\end{align*}
\end{proof}

In the following theorem, which is our first main result,
we denote the spectral family of some self-adjoint
operator, $T$, as $\RR\ni\lambda\mapsto\id_\lambda(T)$.

\begin{theorem}[{\bf Exponential localization}]\label{thm-exp-loc}
Let $e\in\RR$, $\UV>0$, and define $\rho(a)$ by \eqref{def-robert}. 
Then the following assertions hold true:

\smallskip

\noindent
(i) For all $\lambda<\ThPF$, $a\in(0,1)$
with $\ThPF-\lambda>\rho(a)$,
and $\gamma\in(0,\gcPF]$, we have 
$\Ran(\id_\lambda(\PF{\gamma}))\subset\dom(e^{a|\hat{\V{x}}|})$ and
\begin{equation*}
\big\|e^{a|\hat{\V{x}}|}\,\id_\lambda(\PF{\gamma})\big\|
\klg C(a,\lambda,e,\UV)\,.
\end{equation*}
(ii) There is some $c(e,\UV)>0$, such that,
for all $\lambda<\Thnp$, $a\in(0,1)$
satisfying $\Thnp-\rho(a)-c(e,\UV)\,a^2>\lambda$,
and $\gamma\in(0,\gcnp]$, we have 
$\Ran(\id_\lambda(\FNP{\gamma}))\subset\dom(e^{a|\hat{\V{x}}|})$ and
\begin{equation*}
\big\|e^{a|\hat{\V{x}}|}\,\id_\lambda(\FNP{\gamma})\big\|\klg 
C'(a,\lambda,e,\UV)\,.
\end{equation*}
\end{theorem}

\begin{proof}
We treat both models simultaneously again using the notation
\eqref{nora} and the quantities $g$ and $h$ appearing in the
statement of Lemma~\ref{le-HR}.

We put $\Delta:=\Th-\rho(a)-h\,a^2-\lambda$ and
choose $R\grg1\vee(3/\Delta)$ large enough such that
$g/R<\Delta/3$.
Then $H_R\grg\Th-\|V_{\gamma,R}\|_\infty\grg\Th-\Delta/3$;
recall \eqref{nora00}.
Furthermore, we pick some $f\in C_0^\infty(\RR,[0,1])$
satisfying $f=1$ on $[E,\lambda]$ and $f=0$
on $\RR\setminus(E-1,\lambda+\Delta/3)$, so that
$f(H_R)=0$, thus
\begin{equation}\label{BFS}
\chi_R\,\id_\lambda(H)=\big(\chi_R\,f(H)-f(H_R)\,\chi_R\big)\,\id_\lambda(H)\,.
\end{equation}
(This identity with $\chi_R$ replaced by $1$ is observed in \cite{BFS1998b}
for similar purposes.)
As in \cite{BFS1998b} we extend $f$ almost analytically to some
$f\in C^\infty_0(\CC)$ with 
\begin{align*}
\supp(f)&\subset[E-1,\lambda+\Delta/3]+i\,[-1,1]\,,
\\
|\partial_{\ol{z}}f(z)|&\klg C(\Delta,N)\,|\Im z|^N,\quad z\in\CC\,,\;N\in\NN\,,
\end{align*}
and apply the Helffer-Sj\"ostrand formula,
$$
f(T)=\int_\CC(T-z)^{-1}\,d\mu(z)\,,\quad
d\mu(z):=\frac{1}{2\pi i}\,\partial_{\ol{z}}f(z)\,
dz\wedge d\ol{z}\,,
$$
which is valid, for any self-adjoint operator $T$
in some Hilbert space; see, e.g., \cite{DimassiSjoestrand}.
Combining it with \eqref{BFS} and Lemma~\ref{le-res-id1}
we obtain, for every $F\in\sW_{a}$,
\begin{align*}
\chi_R\,e^F\,\id_\lambda(H)
&=\int_\CC e^F\,\big(\chi_R\,(H-z)^{-1}-(H_R-z)^{-1}\chi_R\big)
\,\id_\lambda(H)\,d\mu(z)
\\
&=\int_\CC e^F\,(H_R-z)^{-1}e^{-F}\big(U_{R}^F(z)^*+W_{R}^F(z)^*\big)
\,\id_\lambda(H)\,d\mu(z)\,.
\end{align*}
Applying Lemma~\ref{le-UV}
and Lemma~\ref{le-HR} (with $\delta=\Delta/3$) we arrive at
\begin{equation*}
\sup_{F\in\sW_{a}}\big\|\chi_R\,e^F\,\id_\lambda(H)\big\|\klg 
\frac{C(a,e,\UV,R)}{\Delta}
\int_\CC\frac{|\partial_{\ol{z}}f(z)|}{|\Im z|}\,|dz\wedge d\ol{z}|
\klg C(a,e,\UV,\Delta)\,.
\end{equation*}
To conclude we pick a sequence $F_n\in\sW_{a}$, $n\in\NN$,
converging monotonically to
$a|\V{x}|-a$ on $\{|\V{x}|\grg2\}$. Then, by monotone convergence,
$\int_{\RR^3} e^{2a|\V{x}|}\|\psi(\V{x})\|_{\Fock^4}^2d^3\V{x}
=\lim\limits_{n\to\infty}
\int_{\RR^3} e^{2F_n(\V{x})}\|\psi(\V{x})\|_{\Fock^4}^2d^3\V{x}
\klg C'(a,e,\UV,\Delta)$,
$\psi\in\Ran(\id_\lambda(H))\subset\int_{\RR^3}^\oplus\Fock^4\,d^3\V{x}$. 
\end{proof}

In the non-relativistic setting an analog of
Theorem~\ref{thm-exp-loc} has been obtained in
\cite{BFS1998b,Griesemer2004}.
Let $\Sigma^{\mathrm{nr}}$
denote the ionization threshold of the NRPF Hamiltonian.
Then the decay rate found in \cite{Griesemer2004} for the localization of
states of energy $\klg\lambda<\Sigma^{\mathrm{nr}}$
is $a^2<\Sigma^{\mathrm{nr}}-\lambda$. 
Together with a bound on the increase of binding energy
in the NRPF model (compared to the case $e=0$)
it actually shows that the localization of the lowest energy states
is improved in presence of a quantized radiation field.
In case of the SRPF model the bound on the rate $a$
in Theorem~\ref{thm-exp-loc}(i) is good enough to
demonstrate the same effect.
To explain these issues more precisely we first
consider the case without radiation field.

If we set the parameter $e$ equal to zero and
restrict the operators $\PF{\gamma}$ and $\NP{\gamma}$
to the vacuum sector, then we get back the electronic
operators defined in \eqref{chandra} and \eqref{charly},
respectively. In particular, we may observe the following 
(to the largest part well-known) result:

\begin{corollary}\label{cor-exp-loc-el}
Let $\gc$ be $\gcPF$ or $\gcnp$
and $H_\gamma^\el$ be $H_\gamma^{\el,\Pf}$ or $H_\gamma^{\el,\np}$.
Then, for all $\gamma\in(0,\gc]$,
$\lambda\in[0,1)$, and $a>0$ with $a^2<1-\lambda^2$,
we have $\Ran(\id_\lambda(H_\gamma^\el))\subset\dom(e^{a|\hat{\V{x}}|})$ and
$$
\big\|e^{a|\hat{\V{x}}|}\,\id_\lambda(H_\gamma^\el)\big\|\klg 
C(a,\lambda)\,.
$$ 
\end{corollary}

\begin{proof}
First, we recall that $\ThPF|_{e=0}=\Thnp|_{e=0}=1$.
Hence, if $H_\gamma^\el=H_\gamma^{\el,\Pf}$, then the assertion
of the corollary is contained in the statement of
Theorem~\ref{thm-exp-loc}(i). If $H_\gamma^\el=H_\gamma^{\el,\np}$,
then the assertion of the corollary can be verified by inspection
of the proofs in the present section.
In fact, if we ignore all Fock space operators,
then we may choose $h=0$ in Lemma~\ref{le-HR} also
when we consider the electronic no-pair model.
As a consequence, the constant $c(e,\UV)$ appearing
in the statement of Theorem~\ref{thm-exp-loc}(ii)
can be replaced by zero.
\end{proof}

For the no-pair model the statement of
Corollary~\ref{cor-exp-loc-el} is a special case of a result in
\cite{MatteStockmeyer2008b}, where also non-vanishing classical magnetic fields
are considered. 
For the square root operator the assertion of Corollary~\ref{cor-exp-loc-el}
is well-known, at least for all $\gamma\in(0,\nf{1}{2}]$
\cite{Nardini1986}; see also \cite{CMS1990,HelfferParisse1994}
for exponential decay estimates for square-root operators.
As it seems to us the whole range of allowed $\gamma$
is not covered by the published literature so far.
The bound on the decay rate $a^2<1-\lambda^2$ is familiar from the
analysis of the Dirac operator.

\begin{remark}[On improved localization due to the radiation field]
Let $e_\gamma:=\inf\spec[H_\gamma^{\el,\Pf}]$, $\gamma\in(0,\gcPF]$, 
denote the value of the
lowest eigenvalue of the electronic square root operator.
It is known that
$0<e_\gamma<1$, for all $\gamma\in(0,\gcPF]$.
In fact, strict positivity of $e_{\gcPF}$ 
is shown numerically in \cite{HardekopfSucher1985}
and it is proven analytically in \cite{RRSMS1994}.
Then the value $a_\gamma:=(1-e_\gamma^2)^{\nf{1}{2}}\in(0,1)$
is the border line for all decay rates $a$ a allowed for
in Corollary~\ref{cor-exp-loc-el}. 
Corresponding pointwise 
lower bounds for ground state eigenfunctions
of square root operators (whose potentials belong to a suitable Kato class) 
\cite{CMS1990} suggest $a_\gamma$ to be optimal indeed.
Now, assume that, for $e\not=0$ and $\UV>0$, the {\em binding
energy is increased} in the SRPF model, i.e. assume the {\em strict} inequality
\begin{equation}\label{inke}
\ThPF-E_\gamma^\Pf>1-e_\gamma\,.
\end{equation}
Since $1-e_\gamma=1-(1-a_\gamma^2)^{\nf{1}{2}}=\rho(a_\gamma)$
we conclude by means of
Theorem~\ref{thm-exp-loc}(i) that in this case
$$
\forall\,\gamma\in(0,\gcPF]\;\;\exists\,\ve>0\::\quad
\big\|e^{(a_\gamma+\ve)|\hat{\V{x}}|}\id_{\{E_\gamma\}}(\PF{\gamma})\big\|
\klg C(e,\UV,\ve)\,.
$$
Thus, we observe an enhancement of localization in the ground state
due to the quantized radiation field.
The condition \eqref{inke} will be discussed
by the present authors in a separate paper.
(In the non-relativistic setting it is established in
\cite{CVV2003} under the (implicit) assumption that
$e$ and/or $\UV$ be sufficiently small.)
\end{remark}


\section{GROUND STATES AT CRITICAL COUPLING}\label{sec-cpt}

\noindent
Starting from the assertions of Propositions~\ref{prop-PF-sbc}(ii)
and~\ref{prop-np-sbc}(ii), namely that $\PF{\gamma}$ and $\FNP{\gamma}$
have eigenvalues at the bottom of their spectra, as long
as $\gamma$ is sub-critical, 
we prove in this section
that both operators still possess ground state eigenvectors,
when $\gamma$ attains the critical values $\gcPF$ and $\gcnp$,
respectively.

We shall make use of the following abstract lemma
which is a variant of a result we learned from \cite{BFS1998b};
see \cite[Lemma~5.1]{KMS2009a} for a proof.

\begin{lemma}\label{le-ex-allg}
Let $T,T_1,T_2,\ldots$ be self-adjoint operators acting in some
separable Hilbert space, $\mathscr{X}$, such that $\{T_j\}_{j\in\NN}$
converges to $T$ in the strong resolvent sense.
Assume that $E_j$ is an eigenvalue of $T_j$ with corresponding
eigenvector $\phi_j\in\dom(T_j)$.
If $\{\phi_j\}_{j\in\NN}$ converges weakly to
some $0\not=\phi\in\mathscr{X}$, then $E:=\lim_{j\to\infty}E_j$
exists and is an eigenvalue of $T$. If $E_j=\inf\spec[T_j]$,
then $T$ is semi-bounded below and $E=\inf\spec[T]$.
\end{lemma}

As we wish to consider the limit as $\gamma$
approaches its critical values 
we employ the following new convention from now on:
\begin{equation}\label{nora2}
\left\{
\begin{array}{l}
\textrm{The symbols}\;\,H_\gamma,\,\Th,\,E_\gamma,\,\gc\;\,
\textrm{denote either}\\
\PF{\gamma},\,\Th^\Pf,\,E^\Pf_\gamma,\,\gcPF\;\,
\textrm{or}\;\,\FNP{\gamma},\,\Th^\np,\,E^\np_\gamma,\,\gcnp\,.
\end{array}
\right.
\end{equation}

\begin{lemma}\label{le-srconv}
$H_\gamma$ converges to
$H_{\gc}$ in the strong resolvent sense, 
as $\gamma\nearrow\gc$.
In particular,
\begin{equation}\label{guenter1}
\limsup_{\gamma<\gc} E_\gamma\klg E_{\gc}\,.
\end{equation}
\end{lemma}

\begin{proof}
For every $\gamma\in(0,\gc)$, we know that
$\form(H_\gamma)=\form(|\DO|)\cap\form(\Hf)
\subset\form(H_{\gc})$ \cite{KMS2009b}.
Since $\core$ is a form core for $H_{\gc}$
we thus have 
$\ol{\cap_{\gamma<\gc}\form(H_\gamma)}=\form(H_{\gc})$,
where the closure is taken with respect to the form norm
of $H_{\gc}$. Since the expectation values
$\SPn{\vp}{H_\gamma\,\vp}\searrow\SPn{\vp}{H_{\gc}\,\vp}$
converge monotonically, as $\gamma\nearrow\gc$, for every
$\vp\in \cap_{\gamma<\gc}\form(H_\gamma)=\form(|\DO|)\cap\form(\Hf)$,
it follows from \cite[Satz~9.23a]{Weidmann2000} that
$H_\gamma$ converges to
$H_{\gc}$ in the strong resolvent sense.
\end{proof}

In order to verify the assumption $\phi\not=0$ of Lemma~\ref{le-ex-allg}
we shall adapt a compactness argument from \cite{GLL2001}.
To this end we need
the infra-red bounds of the next proposition which give some
information on the localization and the weak derivatives
of ground state eigenvectors with respect to the photon
variables. In non-relativistic QED soft photon bounds
(without infra-red regularization) have been obtained first 
in \cite{BFS1999} and photon derivative bounds have been introduced
in \cite{GLL2001}. To state these bounds for our models
we recall the notation
$$
(a(k)\,\psi)^{(n)}(k_1,\dots,k_n)\,=\,
(n+1)^{\nf{1}{2}}\,\psi^{(n+1)}(k,k_1,\dots,k_n)\,,\quad n\in\NN_0\,,
$$
almost everywhere, for $\psi=(\psi^{(n)})_{n=0}^\infty\in\Fock[\HP]$,
and $a(k)\,(\psi^{(0)},0,0,\ldots\;)=0$.

\begin{proposition}[{\bf Infra-red bounds}]\label{prop-IR}
Let $e\in\RR$, $\UV>0$, and $\gamma_1\in(0,\gc)$.
Then there is some $C(e,\UV,\gamma_1)\in(0,\infty)$,
such that, for all $\gamma\in[\gamma_1,\gc)$ and every normalized
ground state eigenvector, $\phi_\gamma$, of $H_\gamma$,
we have the {\em soft photon bound},
\begin{equation}\label{eq-spb}
\|a(k)\,\phi_\gamma\|^2\klg\id_{\{|\V{k}|\klg\UV\}}
\:\frac{C(e,\UV,\gamma_1)}{|\V{k}|}\,,
\end{equation}
for almost every $k=(\V{k},\lambda)\in\RR^3\times\ZZ_2$, and
the {\em photon derivative bound},
\begin{equation}\label{pdb-kp}
\big\|a(\V{k},\lambda)\,\phi_\gamma-a(\V{p},\lambda)\,\phi_\gamma\big\|
\klg
C(e,\UV,\gamma_1)\,|\V{k}-\V{p}|\,
\Big(\frac{1}{|\V{k}|^{1/2}|\V{k}_\bot|}+
\frac{1}{|\V{p}|^{1/2}|\V{p}_\bot|}\Big),
\end{equation}
for almost every $\V{k},\V{p}\in\RR^3$
with $0<|\V{k}|<\UV$, $0<|\V{p}|<\UV$, and $\lambda\in\ZZ_2$.
(Here we use the notation \eqref{def-kbot}.)
In particular, 
\begin{equation}\label{eq-spbN}
\sup_{\gamma\in[\gamma_1,\gc)}\sum_{n=1}^\infty n\,\|\phi_\gamma^{(n)}\|^2<\infty\,,
\end{equation}
where
$\phi_\gamma=(\phi_\gamma^{(n)})_{n=0}^\infty\in
\bigoplus_{n=0}^\infty L^2(\RR^3,\CC^4)\otimes\Fock^{(n)}[\HP]$.
\end{proposition}

\begin{proof}
First, we prove the soft photon bound \eqref{eq-spb}
for the SRPF operator. To this end we put
\begin{align*}
\RA{iy}&:=(\DA-iy)^{-1},\;\;\,y\in\RR\,,
\qquad
\cR_\V{k}:=(\PF{\gamma}-E_\gamma^\Pf+|\V{k}|\,)^{-1},
\;\;\,\V{k}\not=0\,,
\end{align*}
and (recall \eqref{def-Gphys})
\begin{align*}
\wt{\V{G}}_\V{x}(k)&:=\V{G}_\V{x}(k)-\V{G}_{\V{0}}(k)=
\V{G}_{\V{0}}(k)\,(e^{-i\V{k}\cdot{\V{x}}}-1)\,.
\end{align*}
For $\gamma\in(0,\gcPF)$, we derived the following
representation in \cite{KMS2009a},
\begin{align*}
a(k)\,\phi_\gamma&:=
i\big(|\V{k}|\,\cR_\V{k}-1\big)\,\V{G}_{\V{0}}(k)\cdot\hat{\V{x}}\phi_\gamma-
\cR_\V{k}\,\valpha\cdot\wt{\V{G}}_{\hat{\V{x}}}(k)
\,\SA\phi_\gamma+I_\gamma(k)\,,
\end{align*}
for almost every $k=(\V{k},\lambda)\in\RR^3\times\ZZ_2$, where
\begin{align*}
I_\gamma(k)&:=\int_\RR\cR_\V{k}\,\DA\,\RA{iy}\,
\valpha\cdot\wt{\V{G}}_{\hat{\V{x}}}(k)
\,\RA{iy}\,\phi_\gamma\,\frac{dy}{\pi}\,.
\end{align*}
Here the Bochner integral $I_\gamma(k)$
is actually absolutely convergent. In fact,
pick some $F\in C^\infty(\RR^3_\V{x},[0,\infty))$
such that $F(\V{x})=a|\V{x}|$, for large $|\V{x}|$,
and $|\nabla F|\klg a$. In view of Theorem~\ref{thm-exp-loc}
we may choose $a\in(0,\nf{1}{2}]$ sufficiently small
(depending on $\gamma_1$) such that
$\sup_{\gamma\in[\gamma_1,\gcPF)}\|e^F\phi_\gamma\|<\infty$.
By virtue of \eqref{iris} and Theorem~\ref{thm-exp-loc}
we then obtain, for all $\gamma\in[\gamma_1,\gcPF)$,
\begin{align*}
\|I_\gamma(k)\|&\klg\int_\RR\Big\{\big\|\,|\DA|^{\nf{1}{4}}\,\cR_\V{k}\,\big\|
\,\big\|\,|\DA|^{\nf{3}{4}}\RA{iy}\big\|
\\
&\quad\cdot
\,|\V{G}_{\V{0}}(k)|\,\sup_\V{x}\big|(e^{-i\V{k}\cdot\V{x}}-1)\,e^{-F(\V{x})}\big|
\,\|e^F\RA{iy}\,e^{-F}\|\,\|e^F\phi_\gamma\|\Big\}\,\frac{dy}{\pi}\,,
\end{align*}
Here $\|\,|\DA|^{\nf{1}{4}}\,\cR_\V{k}\,\|\klg C(e,\UV)/(1\wedge|\V{k}|)$
by \eqref{iris}, 
$\|\,|\DA|^{\nf{3}{4}}\RA{iy}\|\klg C\,\SL y\SR^{-\nf{1}{4}}$,
and, by Lemma~\ref{le-marah}
below, the composition
$e^F\RA{iy}\,e^{-F}$ is well-defined with
$\|e^F\RA{iy}\,e^{-F}\|\klg C\,\SL y\SR^{-1}$. Using also 
$|\V{G}_{\V{0}}(k)|\klg (|e|/2\pi)\,|\V{k}|^{-\nf{1}{2}}\,\id_{\{|\V{k}|\klg\UV\}}$
as well as 
$|e^{-i\V{k}\cdot\V{x}}-1|\klg|\V{k}|\,|\V{x}|$, we arrive at 
\begin{align*}
\|I_\gamma(k)\|&\klg\id_{\{|\V{k}|\klg\UV\}}
\,\frac{C'(e,\UV,\gamma_1)\,|\V{k}|^{\nf{1}{2}}}{1\wedge|\V{k}|}
\cdot\!\!\sup_{\gamma'\in[\gamma_1,\gcPF)}\|e^F\phi_{\gamma'}\|
\klg\id_{\{|\V{k}|\klg\UV\}}\,
\frac{C''(e,\UV,\gamma_1)}{|\V{k}|^{\nf{1}{2}}}\,,
\end{align*}
for almost every $k=(\V{k},\lambda)\in\RR^3\times\ZZ_2$
and $\gamma\in[\gamma_1,\gcPF)$.
Now, it is also clear how to estimate the remaining terms in
the formula for $a(k)\,\phi_\gamma$ and to get \eqref{eq-spb}.
(Notice that 
$\|e^F\,\SA\,\phi_\gamma\|\klg\|e^F\SA\,e^{-F}\|\,\|e^F\phi_\gamma\|$,
where $\|e^F\SA\,e^{-F}\|\klg1+C\,a$ by \eqref{clara2} and a simple
approximation argument.)

In a similar fashion we next
derive the photon derivative bound \eqref{pdb-kp} for the SRPF operator.
In fact,
$\|(\cR_\V{k}-\cR_\V{p})\,\psi\|
\klg|\V{p}|^{-1}|\V{k}-\V{p}|\,\|\cR_\V{k}\,\psi\|$, $\psi\in\HR$, 
by the first resolvent identity, thus
\begin{align*}
\|&I_\gamma(\V{k},\lambda)-I_\gamma(\V{p},\lambda)\|
\\
&\klg
\int_\RR\Big\{\big\|\,|\DA|^{\nf{1}{4}}\,\cR_\V{k}\,\big\|
\,\big\|\,|\DA|^{\nf{3}{4}}\RA{iy}\big\|
\\
&\quad\cdot
\sup_{\V{x}}\big\{|\wt{\V{G}}_\V{x}(\V{k},\lambda)
-\wt{\V{G}}_\V{x}(\V{p},\lambda)|e^{-F(\V{x})}\big\}
\,\|e^F\RA{iy}\,e^{-F}\|\,\|e^F\phi_\gamma\|\Big\}\,\frac{dy}{\pi}
\\
&\quad+\frac{|\V{k}-\V{p}|}{|\V{p}|}\int_\RR\Big\{
\big\|\,|\DA|^{\nf{1}{4}}\,\cR_\V{k}\,\big\|
\,\big\|\,|\DA|^{\nf{3}{4}}\RA{iy}\big\|
\\
&\qquad\qquad\qquad\cdot
\sup_{\V{x}}\big\{|\wt{\V{G}}_\V{x}(\V{p},\lambda)|e^{-F(\V{x})}\big\}
\,\|e^F\RA{iy}\,e^{-F}\|\,\|e^F\phi_\gamma\|\Big\}\,\frac{dy}{\pi}\,.
\end{align*}
Here 
$|\wt{\V{G}}_\V{x}(\V{p},\lambda)|\klg(|e|/2\pi)\,|\V{p}|^{\nf{1}{2}}|\V{x}|
\,\id_{\{|\V{p}|\klg\UV\}}$ and some
elementary estimates \cite{GLL2001}
(see also \cite[\textsection6.3]{KMS2009a})
using the special choice \eqref{pol-vec} of
the polarization vectors reveal that
\begin{align}\label{oh}
\frac{|\wt{\V{G}}_\V{x}(\V{k},\lambda)
-\wt{\V{G}}_\V{x}(\V{p},\lambda)|}{|\V{k}|}
\klg C\,(1+|\V{x}|)\,|\V{k}-\V{p}|
\,\Big(\frac{1}{|\V{k}|^{1/2}|\V{k}_\bot|}+
\frac{1}{|\V{p}|^{1/2}|\V{p}_\bot|}\Big)\,,
\end{align}
provided that $0<|\V{k}|,|\V{p}|<\UV$.
By Young's inequality, also 
$|\V{k}-\V{p}|\,|\V{k}|^{-1}|\V{p}|^{-\nf{1}{2}}$
is bounded by the RHS of \eqref{oh}.
Putting these remarks together we conclude that 
$\|I_\gamma(\V{k},\lambda)-I_\gamma(\V{p},\lambda)\|$
is bounded from above by the RHS of \eqref{pdb-kp},
for $0<|\V{k}|,|\V{p}|<\UV$.
Again we leave the treatment of the first two
terms in the formula for $a(k)\,\phi_\gamma$ to the reader;
we just note that 
$|\V{k}|^{-1}\big|\,|\V{k}|\V{G}_\V{0}(\V{k},\lambda)
-|\V{p}|\V{G}_\V{0}(\V{k},\lambda)\big|$ can bounded by the
RHS of \eqref{oh}, too; 
see \cite{GLL2001} or \cite[\textsection6.3]{KMS2009a}.

Finally, in the case of the no-pair operator
we already observed in \cite[Remark~7.2]{KMS2009b}
that the bound proven in Theorem~\ref{thm-exp-loc}(ii)
provides a proof of the infra-red bounds \eqref{eq-spb}
and \eqref{pdb-kp} with a constant independent of $\gamma\in[\gamma_1,\gcnp)$.
In fact, in \cite{KMS2009b}
we derived a formula for $a(k)\,\phi_\gamma$, when
$\phi_\gamma$ is a ground state eigenvector of $\FNP{\gamma}$,
$\gamma\in(0,\gc)$,
which comprises of more terms than in the SRPF case
but is otherwise completely analogous.
Hence, by essentially the same estimates as above we may
derive the infra-red bounds also for the no-pair model.
\end{proof}

Finally, we arrive at the principal result of this article:

\begin{theorem}[{\bf Ground states at critical coupling}]\label{thm-ex-crit}
For $e\in\RR$ and $\UV>0$, the minima of the spectra
of both $\PF{\gcPF}$ and $\NP{\gcnp}$ are eigenvalues.
\end{theorem}

\begin{proof}
Again we treat both models simultaneously using the 
notation \eqref{nora2}.
(Recall that in view of \eqref{NPFNP} it suffices
to show the existence of ground states for
$\FNP{\gcnp}$ instead of $\NP{\gcnp}$ in the no-pair model.)

Let $\phi_\gamma$ denote a normalized
ground state eigenvector of $H_{\gamma}$, for every $\gamma\in(0,\gc)$.
Then the family $\{\phi_\gamma\}_{\gamma\in(0,\gc)}$ contains some weakly
convergent sequence, $\{\phi_{\gamma_j}\}_{j\in\NN}$, $\gamma_j\nearrow\gc$.
We denote the weak limit of the latter by $\phi_{\gc}$. 
On account of Lemmata~\ref{le-ex-allg} and~\ref{le-srconv}
it suffices to show that $\phi_{\gc}\not=0$.

With the exponential localization
and infra-red bounds at hand the following compactness
argument is the same as in \cite{GLL2001}
(where an artificial photon mass is removed instead),
except that we first take the partial Fourier transform
with respect to $\V{x}$ before we apply the
Rellich-Kondrashov theorem.
(If one does not exchange the roles of the electronic position 
and momentum coordinates then the compactness argument
requires imbedding theorems for more exotic
function spaces since one has to deal with fractional
derivaties w.r.t. $\V{x}$ \cite{KMS2009a,KMS2009b}.
The variant of the argument below can also
be used to simplify the proofs in \cite{KMS2009a,KMS2009b}.)

Let $\ve>0$.
On account of \eqref{eq-spbN} we find some $n_0\in\NN$
such that
\begin{equation}\label{don1}
\forall\,\gamma\in[\gamma_1,\gc)\,:\quad
\sum_{n=n_0+1}^\infty\|\phi_\gamma^{(n)}\|^2<\frac{\ve}{2}\,.
\end{equation}
For $n\in\NN$, $\gamma\in(0,\gc]$, and 
$\ul{\theta}
=(\vs,\lambda_1,\dots,\lambda_n)\in\{1,2,3,4\}\times\ZZ_2^n$,
we set 
$$
\phi_{\gamma,\ul{\theta}}^{(n)}(\V{x},\V{k}_1,\dots,\V{k}_n)
:=\phi_\gamma^{(n)}(\V{x},\vs,\V{k}_1,\lambda_1,\ldots,\V{k}_n,\lambda_n)
$$
and denote the partial Fourier transform of $\phi_{\gamma,\ul{\theta}}^{(n)}$
with respect to $\V{x}$ as $\hat{\phi}_{\gamma,\ul{\theta}}^{(n)}$.
Then the soft photon bound \eqref{eq-spb} shows that
$\hat{\phi}_{\gamma,\ul{\theta}}^{(n)}(\V{\vxi},\V{k}_1,\dots,\V{k}_n)=0$, 
for almost every
$(\V{\vxi},\V{k}_1,\dots,\V{k}_n)\in\RR^{3(n+1)}$,
such that $|\V{k}_j|>\Lambda$, for some $j\in\{1,\dots,n\}$.
Moreover, pick some $s\in(0,1)$.
By virtue of \eqref{iris} we then have, for all $\gamma\in(0,\gc)$,
$n\in\NN$, and every choice of $\ul{\theta}$,
\begin{align*}
R^{s}&\int_{|\vxi|\grg R}
\|\hat{\phi}_{\gamma,\ul{\theta}}^{(n)}(\vxi,\cdot)\|^2_{L^2(\RR^{3n})}\,d^3\vxi
\klg\SPn{\phi_{\gamma,\ul{\theta}}^{(n)}}{(-\Delta)^{s/2}\,
\phi_{\gamma,\ul{\theta}}^{(n)}}
\\
&\klg
\SPn{\phi_\gamma}{H_{\gamma}\,\phi_\gamma}+C(e,\UV,s)=
E_\gamma+C(e,\UV,s)
\klg
|E_{\gc}|+\Th+C(e,\UV,s)\,.
\end{align*}
Consequently, we find some $R\grg1$ such that
\begin{equation}\label{don2}
\sum_{n=1}^{n_0}\big\|\id_{\{|\vxi|\grg R\}}\,
\hat{\phi}_{\gamma}^{(n)}\big\|^2<\frac{\ve}{2}\,.
\end{equation}
As in \cite{GLL2001} 
an application of H\"older's
inequality with respect to $d^3\V{\vxi}\,d^{3(n-1)}\V{K}$
and the photon derivative bound \eqref{pdb-kp}
yield, for 
$p\in[1,2)$ and $\gamma\in[\gamma_1,\gc)$,
\begin{align*}
\int\limits_{{|\V{k}|<\UV,\atop |\V{k}+\V{h}|<\UV}}\!\!\!\!\!\iint&
\big|\hat{\phi}_{\gamma,\ul{\theta}}^{(n)}
({\vxi},\V{k}+\V{h},\V{K})
-\hat{\phi}_{\gamma,\ul{\theta}}^{(n)}(\vxi,\V{k},\V{K})\big|^p
\,d^3\vxi\,d^{3(n-1)}\V{K}\,d^3\V{k}
\\
&\klg\,
C\sum_{\lambda\in\ZZ_2}\int\limits_{{|\V{k}|<\UV,\atop |\V{k}+\V{h}|<\UV}}\!\!\!
\big\|\,a(\V{k}+\V{h},\lambda)\,\phi_\gamma
-a(\V{k},\lambda)\,\phi_\gamma\,\big\|^p
d^3\V{k}
\\
&\klg\,
C'\,|\V{h}|^p\!\!\!\int\limits_{|(u,v)|<\UV}\!\!\!
\bigg\{\!
\int\limits_0^{|(u,v)|}\!\frac{dr}{|(u,v)|^{\nf{p}{2}}}\,+
\!\!\!
\int\limits_{|(u,v)|}^\UV\frac{dr}{r^{\nf{p}{2}}}
\bigg\}\,\frac{du\,dv}{|(u,v)|^p}
\,=\,C''\,|\V{h}|^p,
\end{align*}
where the constants $C,C',C''\in(0,\infty)$ 
depend on $p,n,e,\UV$, but not on $\gamma\in[\gamma_1,\gc)$.
Since $\phi_\gamma^{(n)}$ is permutation
symmetric with respect to the variables $k_1,\ldots,k_n$
the previous estimate implies \cite[\textsection4.8]{Nikolskii}
that the weak first order partial derivatives of 
$\hat{\phi}_{\gamma,\ul{\theta}}^{(n)}$ 
with respect to its last $3n$ variables exist on 
$Q_n:=B_R\times B_\UV^n$,
where $B_\rho$ denotes the open ball in $\RR^3$ of radius $\rho$
centered at the origin,
and that
$$
\sup_{\gamma\in[\gamma_1,\gc)}
\|\nabla_{\V{k}_i}\hat{\phi}_{\gamma,\ul{\theta}}^{(n)}\|_{L^p(Q_{n})}
<\infty\,, \quad p\in[1,2)\,,\quad i=1,\ldots,n\,,\quad n=1,\ldots,n_0\,.
$$
Finally, since 
$\sup_{\gamma\in[\gamma_1,\gc)}\|e^{a|\hat{\V{x}}|}\phi_{\gamma,\ul{\theta}}^{(n)}\|
<\infty$,
for some $a>0$,
we know that $\hat{\phi}_{\gamma,\ul{\theta}}^{(n)}$ has 
weak first order derivatives with respect to $\vxi$ and,
for all $\gamma\in[\gamma_1,\gc)$, we have
\begin{align*}
\|\nabla_{\vxi}\hat{\phi}_{\gamma,\ul{\theta}}^{(n)}\|_{L^p(Q_{n})}
&\klg C(p,n,R,\UV)\,\|\nabla_{\vxi}\hat{\phi}_{\gamma,\ul{\theta}}^{(n)}
\|_{L^2(\RR^{3(n+1)})}
\\
&= C'(p,n,R,\UV)\,
\|\hat{\V{x}}\,{\phi}_{\gamma,\ul{\theta}}^{(n)}\|_{L^2(\RR^{3(n+1)})}
\klg C''(p,n,R,\UV)\,.
\end{align*}
As observed in \cite{GLL2001} bounds
with respect to the $L^p$-norms, $p<2$, are actually sufficient
in this situation. In fact,
if we choose $p\in[1,2)$ so large that
$2<\frac{3(n_0+1)\,p}{3(n_0+1)-p}$,
then, for every $n=1,\ldots,n_0$ and every choice of $\ul{\theta}$,
we may apply the Rellich-Kondrashov theorem to show that
every subsequence of
$\{\hat{\phi}_{\gamma_j,\ul{\theta}}^{(n)}\}_{j\in\NN}$
contains another subsequence which is strongly convergent
in $L^2(Q_{n})$.
(Obviously, $Q_n$ satisfies the required cone condition.)
By finitely many repeated selections of subsequences 
we may hence assume without loss of generality
that $\{\hat{\phi}_{\gamma_j,\ul{\theta}}^{(n)}\}_{j\in\NN}$ converges strongly
in $L^2(Q_{n})$ to $\hat{\phi}^{(n)}_{\gc,\ul{\theta}}$, 
for all $n=0,\ldots,n_0$ and $\ul{\theta}$.
Taking \eqref{don1} and \eqref{don2} into account we arrive at
$$
\|\phi_{\gc}\|^2=\sum_{n=0}^\infty\|\hat{\phi}_{\gc}^{(n)}\|^2\grg
\lim_{j\to\infty}
\sum_{n=0}^{n_0}\sum_{\ul{\theta}}
\|\hat{\phi}^{(n)}_{\gamma_j,\ul{\theta}}\|_{L^2(Q_{n})}^2
\grg\lim_{j\to\infty}\|\phi_{\gamma_j}\|^2-\ve=1-\ve\,.
$$
Since $\ve>0$ is arbitrary we conclude that
$\|\phi_{\gc}\|=1$.
\end{proof}


\appendix

\section{ESTIMATES ON COMMUTATORS}\label{app-comm}

\noindent
In this appendix we derive some bounds on the operator
norms of certain commutators involving the sign
function of the Dirac operator which have been used
repeatedly in the main text.
Except for those of Lemma~\ref{le-sandra1}
all results and estimations presented here are
variants of earlier ones in \cite{MatteStockmeyer2009a}.
Nevertheless, we shall give a self-contained exposition 
for the convenience of the reader.

The following basic lemma, stating that the resolvent
of the Dirac operator,
$$
\RA{iy}:=(\DA-iy)^{-1},\qquad y\in\RR\,,
$$
stays bounded after conjugation with suitable exponential weights,
is more or less
folkloric, at least in the case of classical vector
potentials. The proof of \eqref{marah1} given, e.g.,
in \cite{MatteStockmeyer2008b} for classical vector
potentials works for quantized ones without any changes.

\begin{lemma}\label{le-marah}
Let $y\in \RR$, $a\in[0,1)$, and 
$F\in C^\infty(\RR_{\V{x}}^3,\RR)$  
such that $|\nabla F|\klg a$. Then
$iy\in\vr(D_{\V{A}}+i\valpha\cdot\nabla F)$ and
\begin{align}\label{marah0}
R^F_{\V{A}}(iy)&:= e^{F}\,R_{\V{A}}(iy)\,e^{-F}
=(\DA+i\valpha\cdot\nabla F-iy)^{-1}\quad \textrm{on}\;\dom(e^{-F})\,,
\\\label{marah1}
\|R^F_{\V{A}}(iy)\|&\klg \sqrt{6}\,(1-a^2)^{-1}\SL y\SR^{-1}.
\end{align}
\end{lemma}

The factor $(1-a^2)^{-1}$ in \eqref{marah1} will enter into many
estimates below but most of the time we will absorb it into
some constant. Henceforth, we stick to the convention
that all constants $C(a,\ldots),C'(a,\ldots),\ldots\,$
be increasing functions of $a$ when the other displayed
parameters are kept fixed.

All commutator estimates below are based on the 
following representation of $\SA=\DA\,|\DA|^{-1}$ as a 
strongly convergent principal value,
\begin{equation}\label{sgn}
\SA\,\psi=\lim_{\tau\to\infty}\int_{-\tau}^\tau\RA{iy}\,\psi\,\frac{dy}{\pi}\,,
\qquad\psi\in\HR.
\end{equation}

\begin{lemma}\label{le-sandra1}
For every bounded $F\in C^\infty(\RR^3_\V{x},\RR)$
with $|\nabla F|\klg a<1$, all $\chi\in C^\infty(\RR^3_\V{x},[0,1])$
with bounded first order derivatives,
and $\kappa\in[0,1)$,
$$
\big\|\,|\hat{\V{x}}|^{-\kappa}\,(\Hf+1)^{-\nicefrac{1}{2}}\,
[e^{F}\,\SA\,e^{-F},\,\chi]\,\big\|\klg 
C(a,e,\UV,\kappa)\,\|\nabla\chi\|_\infty\,.
$$
\end{lemma}

\begin{proof}
To begin with we put $\HT:=\Hf+1$ and observe that
\begin{equation}\label{annalisa}
\HT^{-\nf{1}{2}}\,\RAF{iy}=\RO{iy}
\,\big(\HT^{-\nf{1}{2}}-T\,\RAF{iy}\big)\,,
\end{equation}
where $T\in\LO(\HR)$ is the closure of
$\HT^{-\nf{1}{2}}\valpha\cdot(\V{A}+i\nabla F)$
and satisfies $\|T\|\klg C(e,\UV)$.
In fact, since $\RO{iy}$ and $\HT^{-\nf{1}{2}}$ commute we obtain,
for every $\vp\in \core$,
\begin{align*}
\big\{&\HT^{-\nf{1}{2}}\,\RAF{iy}-\RO{iy}\,\HT^{-\nf{1}{2}}\big\}
\,(\DA+i\valpha\cdot\nabla F-iy)\,\vp
\\
&=-\RO{iy}\,\HT^{-\nf{1}{2}}\valpha\cdot(\V{A}+i\nabla F)\,\vp
\\
&=-\RO{iy}\,T\,\RAF{iy}\,(\DA+i\valpha\cdot\nabla F-iy)\,\vp\,.
\end{align*}
As $\DA$ is essentially self-adjoint on $\core$
we know that 
$(\DA+i\valpha\cdot\nabla F-iy)\,\core$ is dense in $\HR$
and we obtain \eqref{annalisa}.
(In fact, if $\psi\in\HR$ and $\vp_n\in\core$ converge to
$\RAF{iy}\,\psi\in\dom(\DA)$ in the graph norm of $\DA-iy$,
then $(\DA+i\valpha\cdot\nabla F-iy)\,\vp_n\to\psi$.)
Applying the generalized Hardy inequality, 
$|\hat{\V{x}}|^{-2\kappa}\klg C(\kappa)\,|\DO|^{2\kappa}$,
and $\|\,|\DO|^{\kappa}\RO{iy}\|\klg C'(\kappa)\SL y\SR^{\kappa-1}$
we deduce from \eqref{marah1} and \eqref{annalisa} that
$$
\big\|\,|\hat{\V{x}}|^{-\kappa}\,\HT^{-\nf{1}{2}}\,\RAF{iy}\big\|
\klg C''(e,\UV,\kappa)\,(1-a^2)^{-1}\,\SL y\SR^{\kappa-1}.
$$
Together with \eqref{sgn},
$[\RAF{iy},\chi]=\RAF{iy}\,i\valpha\cdot\nabla\chi\,\RAF{iy}$,
and \eqref{marah0}\&\eqref{marah1}
this permits to get
\begin{align*}
\big|&\SPb{|\hat{\V{x}}|^{-\kappa}\,\vp}{\HT^{-\nicefrac{1}{2}}\,
[e^{F}\,\SA\,e^{-F},\,\chi]\,\psi}\big|
\\
&\klg\int_\RR
\big|\SPb{|\hat{\V{x}}|^{-\kappa}\,\vp}{\HT^{-\nicefrac{1}{2}}\,
\RAF{iy}\,i\valpha\cdot\nabla\chi\,\RAF{iy}\,\psi}\big|
\,\frac{dy}{\pi}
\\
&\klg C'''(e,\UV,\kappa)\,(1-a^2)^{-2}
\int_\RR\SL y\SR^{\kappa-2}dy\cdot\|\nabla\chi\|_\infty
\,\|\vp\|\,\|\psi\|\,,
\end{align*}
for all $\vp\in\dom(|\hat{\V{x}}|^{-\kappa})$, $\psi\in\HR$, and we conclude.
\end{proof}

The bounds derived in the following lemma
are slightly more general than the corresponding ones of
\cite[Lemma~3.5]{MatteStockmeyer2009a}.

\begin{lemma}\label{le-sandra2}
Let $\kappa\in[0,1)$, $\ve>0$,
and $\chi\in C^\infty(\RR^3_{\V{x}},[0,1])$ with
$|\nabla\chi|$ bounded.
Moreover, let $F,G\in C^\infty(\RR^3_\V{x},\RR)$ be bounded
with bounded first order derivatives 
and such that $|\nabla(F-G)|\klg a<1$.
Then 
\begin{align}
\label{sandra2.2}
\big\|\,|\DA|^\kappa\,
[\chi\,e^G\,,\,\SA\,]\,e^{F-G}\,\big\|&\klg 
C(a,\kappa)\,\|(\nabla\chi+\chi\,\nabla G)\,e^F\|_\infty\,,
\\
\label{sandra2.1}
\big\|\,|\DA|^{-\ve}\,
[\chi\,e^G\,,\,|\DA|\,]\,e^{F-G}\,\big\|&\klg 
C(a,\ve)\,\|(\nabla\chi+\chi\,\nabla G)\,e^F\|_\infty\,.
\end{align}
In particular, we have, for every bounded 
$G\in C^\infty(\RR^3_\V{x},\RR)$ 
such that $|\nabla G|\klg a<1$,
\begin{equation}\label{clara2}
\big\|e^G\SA e^{-G}\big\|\klg 1+C(a)\,\|\nabla G\|_\infty\,.
\end{equation}
\end{lemma}

\begin{proof}
Combining \eqref{sgn}, the computation
\begin{equation}\label{sandra-1}
[\RA{iy},\chi\,e^G]\,e^{F-G}=
\RA{iy}\,i\valpha\cdot(\nabla\chi+\chi\,\nabla G)\,e^F\,
R_\V{A}^{G-F}(iy)\,,
\end{equation}
and the bounds 
$\|\,|\DA|^{\kappa}\,\RA{iy}\|\klg C'(\kappa)\,\SL y\SR^{\kappa-1}$
and $\|R_\V{A}^{G-F}(iy)\|\klg C'(a)\,\SL y\SR^{-1}$
we find, for all $\vp\in\dom(|\DA|^{\kappa})$ and $\psi\in\HR$,
\begin{align}\nonumber
\big|&\SPb{|\DA|^{\kappa}\,\vp}{[\chi\,e^G,\SA]\,e^{F-G}\,\psi}\big|
\\\nonumber
&\klg\int_\RR
\big|\SPb{|\DA|^{\kappa}\,\vp}{
\RA{iy}\,
i\valpha\cdot(\nabla\chi+\chi\,\nabla G)\,e^F\,
R_\V{A}^{G-F}(iy)
\,\psi}\big|\,\frac{dy}{\pi}
\\\label{sandra2.3}
&\klg C''(a,\kappa)\,\|(\nabla\chi+\chi\,\nabla G)\,e^F\|_\infty
\int_\RR\SL y\SR^{\kappa-2}dy\,\|\vp\|\,\|\psi\|\,,
\end{align}
which gives \eqref{sandra2.2}.
Choosing $\kappa=0$, $\chi=1$, and $F=0$ we also obtain \eqref{clara2},
$$
\big\|e^G\SA e^{-G}\big\|\klg 
\|\SA\|+\big\|\,[e^G,\SA]\,e^{-G}\|\klg1+C(a)\,\|\nabla G\|_\infty\,.
$$ 
To derive \eqref{sandra2.1} we write $|\DA|=\DA\,\SA$
and compute
$$
[\chi\,e^G,|\DA|\,]\,e^{F-G}=
i\valpha\cdot(\nabla\chi+\chi\,\nabla G)\,e^F\,\big(e^{G-F}\SA
e^{F-G}\big)+\DA\,[\chi\,e^G,\SA]\,e^{F-G}
$$
on $\core$.
(Thanks to \cite[Proof of Lemma~3.4(ii)]{MatteStockmeyer2009a}
we know that $\SA$ maps $e^{F-G}\core=\core$ into
$\dom(\DO)\cap\dom(\Hf)$ which is left invariant under
multiplication with $\chi\,e^G$.)
Using $|\DA|^{-\ve}\DA=\SA\,|\DA|^\kappa$ with $\kappa:=1-\ve<1$
we thus observe that \eqref{sandra2.1} is a consequence
of \eqref{sandra2.2} and \eqref{clara2}.
\end{proof}

The next lemma again presents
a variant of a bound from \cite[Lemma~3.5]{MatteStockmeyer2009a}.
In order to prove it we recall some technical
tool introduced in \cite{MatteStockmeyer2009a}.
First, we put 
\begin{equation}\label{def-HT}
\HT:=\Hf+K\,,\qquad
T_\nu:=[\HT^{-\nu},\valpha\cdot\V{A}]\,\HT^\nu\;\;\textrm{on}\;\core,
\end{equation}
and recall the bound 
$\|T_\nu\|\klg C(e,\UV)/K^{\nf{1}{2}}$, for $\nu\grg1/2$
and $K\grg 1$;
see \cite[Lemma~3.1]{MatteStockmeyer2009a}.
In view of \eqref{marah1}
it shows that, for a sufficiently large choice of $K\grg1$,
the Neumann series 
$\Xi^F_\nu(y):=\sum_{\ell=0}^\infty\{-\RAF{iy}\,\ol{T}_\nu\}^\ell$ converges 
and satisfies, say, $\|\Xi^F_\nu(y)\|\klg2$, for all
$\nu\grg1/2$, $y\in\RR$, and $F\in C^\infty(\RR^3_\V{x},\RR)$
with $|\nabla F|\klg a<1$.
Moreover, it is easy to verify the following
useful intertwining relation \cite[Corollary~3.1]{MatteStockmeyer2009a},
\begin{align}\label{Com1}
\HT^{-\nu}\,R^F_{\V{A}}(iy)
&=\Xi_{\nu}^F(y)\,R^F_{\V{A}}(iy)\,\HT^{-\nu}.
\end{align}

\begin{lemma}\label{le-sandra3}
Let $\nu\grg1/2$ and $\chi$, $F$, and $G$
be as in Lemma~\ref{le-sandra2}. Then
\begin{equation}\label{clara3}
\big\|(\Hf+1)^{-\nu}\,[\chi\,e^G\,,\,\SA]\,e^{F-G}\,\Hf^\nu\,\big\|
\klg C(a,e,\UV)^\nu\,\|(\nabla\chi+\chi\,\nabla G)\,e^F\|_\infty\,.
\end{equation}
\end{lemma}

\begin{proof}
We define $\HT$ by \eqref{def-HT}, for some sufficiently large
$K\grg1$ such that the remarks preceding the statement are
applicable.
By means of \eqref{sgn}, \eqref{sandra-1}, and \eqref{Com1} we then obtain
\begin{align}\nonumber
\big|&\SPb{\vp}{\HT^{-\nu}\,[\chi\,e^G,\SA]\,e^{F-G}\,\Hf^\nu\,\psi}\big|
\\\nonumber
&\klg\int_\RR
\big|\SPb{\vp}{\HT^{-\nu}\,
\RA{iy}\,i\valpha\cdot(\nabla\chi+\chi\nabla G)\,e^F\,
R_\V{A}^{G-F}(iy)\,\Hf^\nu\,\psi}\big|\,
\frac{dy}{\pi}
\\\nonumber
&\klg\int_\RR
\big|\SPb{\vp}{\Xi^0_\nu(y)\,
\RA{iy}\,i\valpha\cdot(\nabla\chi+\chi\nabla G)\,e^F\,\times
\\
&\nonumber
\qquad\qquad\qquad\qquad\qquad\qquad\quad\times\,\Xi^{G-F}_\nu(y)\,
R_\V{A}^{G-F}(iy)\,\HT^{-\nu}\,\Hf^\nu\,\psi}\big|\,\frac{dy}{\pi}
\\\nonumber
&\klg C(a)\,\sup_{y\in\RR}\{\|\Xi^0_\nu(y)\|\,\|\Xi^{G-F}_\nu(y)\|\}
\,\|\Hf^\nu\,\HT^{-\nu}\|
\,\|(\nabla\chi+\chi\,\nabla G)\,e^F\|_\infty
\int_\RR\SL y\SR^{-2}dy
\\\nonumber
&\klg 
C'(a)\,\|(\nabla\chi+\chi\,\nabla G)\,e^F\|_\infty\,,
\end{align}
for all normalized $\vp,\psi\in\core$.
This implies \eqref{clara3} since
$\|(\Hf+1)^{-\nu}\HT^{\nu}\|\klg K^\nu$, where our choice
of $K$ depends only on $a$, $e$, and $\UV$.
\end{proof}

The last lemma of this appendix is just a special case
of \cite[Lemma~3.6]{MatteStockmeyer2009a}.

\begin{lemma}\label{le-dc}
For all bounded $F\in C^\infty(\RR^3_\V{x},\RR)$ such that
$|\nabla F|\klg a<1$ and $\nu\grg1/2$,
\begin{equation}\label{clara99}
\big\|(\Hf+1)^{-\nu}\,\big[e^{-F},\,[\SA\,,\,e^F]\big]\,\Hf^\nu\,\big\|
\klg C(a,e,\UV)^\nu\,\|\nabla F\|_\infty^2\,.
\end{equation}
\end{lemma}

\begin{proof}
A straightforward computation yields
$$
\big[e^{-F},\,[\RA{iy}\,,\,e^F]\big]
=\RA{iy}\,i\valpha\cdot\nabla F\,
\big\{\RAF{iy}+\RAmF{iy}\big\}\,i\valpha\cdot\nabla F\,\RA{iy}\,.
$$
Together with \eqref{marah1}, \eqref{sgn}, 
and \eqref{Com1} this permits to get
\begin{align*}
\big|&\SPb{\vp}{\HT^{-\nu}\big[e^{-F},\,[\SA\,,\,e^F]\big]\,\Hf^\nu\,\psi}\big|
\\
&\klg2^3\int_\RR
\|\RA{iy}\|\,\|\nabla F\|_\infty^2\,
\big(\|\RAF{iy}\|+\|\RAmF{iy}\|\big)\,\|\RA{iy}\|\,\|\Hf^\nu\,\HT^{-\nu}\|\,
\frac{dy}{\pi}
\\
&\klg C(a)\,\|\nabla F\|_\infty^2\int_\RR\frac{dy}{\SL y\SR^3}\,,
\end{align*}
for all normalized $\vp,\psi\in\core$. We conclude 
as in the previous proof. 
\end{proof}


\bigskip

\noindent
{\bf Acknowledgement.}
This work has been partially supported by the DFG (SFB/TR12).
We thank the Erwin Schr\"odinger Institute for Mathematical Physics in Vienna,
where parts of this work have been prepared, 
for their kind hospitality.
Finally, we thank the anonymous referee for a useful remark on
the enhancement of localization due to the radiation field.


\vspace{1.5cm}

{\footnotesize

\noindent
{\sc Martin K\"onenberg}\\ 
Fakult\"at f\"ur Mathematik und Informatik\\
FernUniversit\"at Hagen\\
L\"utzowstra{\ss}e 125\\ 
58084 Hagen, Germany.\\
{\em Present address:}\\ 
Fakult\"at f\"ur Physik\\
Universit\"at Wien\\
Boltzmanngasse 5\\
1090 Vienna, Austria.\\
{\tt martin.koenenberg@univie.ac.at}
 
\bigskip

\noindent
{\sc Oliver Matte}\\ 
Fakult\"at f\"ur Mathematik\\
Technische Universit\"at M\"unchen\\
Boltzmannstra{\ss}e 3\\
85748 Garching, Germany.\\ 
{\em Present address:}\\ 
Mathematisches Institut\\
Ludwig-Maximilians-Universit\"at\\
Theresienstra{\ss}e 39\\
80333 M\"unchen, Germany.\\
{\tt matte@math.lmu.de}
}

\end{document}